\documentclass[american,aps,prx,reprint,floatfix,nofootinbib,superscriptaddress,longbibliography]{revtex4-2}
\usepackage{algorithm}
\usepackage{algpseudocode}
\usepackage{graphicx} % Required for inserting images
\usepackage{braket}
\usepackage{amsmath}
\usepackage{amsfonts,amsthm}
\usepackage{xcolor}
\usepackage{listings}
\usepackage{dsfont}
\usepackage{verbatim}
\usepackage[bb=libus]{mathalpha}
\usepackage{comment}

\DeclareMathOperator{\Tr}{Tr}
\newcommand{\ra}{\rangle}
\newcommand{\la}{\langle}
\newcommand{\vc}{\vec c}
\renewcommand{\oc}{\vec{c}^{\,\star}}

\newcommand{\calE}{{\cal E }}

\newcommand{\calO}{{\cal O }}

\newcommand{\one}{\mathds{1}}

\newcommand{\mpoM}{M^{(\ell_0)}}
\newcommand{\mpoL}{L^{(\ell_0)}}
\newcommand{\mpoMerr}{\calE_M^{(\ell_0)}}
\newcommand{\mpoLerr}{\calE_L^{(\ell_0)}}
\newcommand{\mpoC}{\vc^{\,(\ell_0)}}

\newtheorem{lemma}[section]{Lemma}

\newcommand{\beq}{\begin{equation}}
\newcommand{\eeq}{\end{equation}}

\begin{document}

\title{Tensor Network enhanced Dynamic Multiproduct Formulas}

\author{Niall F. Robertson}
\thanks{co-first author}
\affiliation{IBM Quantum, IBM Research Europe - Dublin, IBM Technology Campus, Dublin 15, Ireland}

\author{Bibek Pokharel}
\thanks{co-first author}
\affiliation{IBM Quantum, IBM Research -- Almaden, San Jose CA, 95120, USA}

\author{Bryce Fuller}
\affiliation{IBM Quantum, IBM Thomas J Watson Research Center, Yorktown Heights, NY 10598, USA}

\author{Eric Switzer}
\affiliation{Donostia International Physics Center (DIPC), 20018 Donostia-San Sebastian, Euskadi, Spain}
\affiliation{Department of Physics, University of Central Florida, Orlando, Florida 32816, USA}

\author{Oles Shtanko}
\affiliation{IBM Quantum, IBM Research -- Almaden, San Jose CA, 95120, USA}

\author{Mirko Amico}
\affiliation{IBM Quantum, IBM Thomas J Watson Research Center, Yorktown Heights, NY 10598, USA}

\author{Adam Byrne}
\affiliation{IBM Quantum, IBM Research Europe - Dublin, IBM Technology Campus, Dublin 15, Ireland}

\author{Andrea D'Urbano}
\affiliation{IBM Quantum, IBM Research Europe - Dublin, IBM Technology Campus, Dublin 15, Ireland}

\author{Salome Hayes-Shuptar}
\affiliation{IBM Quantum, IBM Research Europe - Dublin, IBM Technology Campus, Dublin 15, Ireland}

\author{Albert Akhriev}
\affiliation{IBM Quantum, IBM Research Europe - Dublin, IBM Technology Campus, Dublin 15, Ireland}

\author{Nathan Keenan}
\affiliation{IBM Quantum, IBM Research Europe - Dublin, IBM Technology Campus, Dublin 15, Ireland}

\author{Sergey Bravyi}
\affiliation{IBM Quantum, IBM Thomas J Watson Research Center, Yorktown Heights, NY 10598, USA}

\author{Sergiy Zhuk}
\affiliation{IBM Quantum, IBM Research Europe - Dublin, IBM Technology Campus, Dublin 15, Ireland}

\begin{abstract}
Tensor networks and quantum computation are two of the most powerful tools for the simulation of quantum many-body systems. Rather than viewing them as competing approaches, here we consider how these two methods can work in tandem. We introduce a novel algorithm that combines tensor networks and quantum computation to produce results that are more accurate than what could be achieved by either method used in isolation. Our algorithm is based on multiproduct formulas (MPF) - a technique that linearly combines Trotter product formulas to reduce algorithmic error. Our algorithm uses a quantum computer to calculate the expectation values and tensor networks to calculate the coefficients used in the linear combination. We present a detailed error analysis of the algorithm and demonstrate the full workflow on a one-dimensional quantum simulation problem on $50$ qubits using two IBM quantum computers: \texttt{ibm\char`_torino} and \texttt{ibm\char`_kyiv}.
\end{abstract}

\maketitle

\section{Introduction}

Understanding out-of-equilibrium properties of quantum systems relevant for chemistry, material science, and high-energy physics often requires the simulation of Hamiltonian dynamics. For example, correlation functions 
describing unitary time evolution  of  interacting quantum spins or electrons  provide information about the excitation spectrum and aid identification  of exotic quasiparticles such as unpaired Majorana fermions in one-dimensional~\cite{harle2022observing} and two-dimensional~\cite{knolle2014dynamics} models. 
Correlation functions combining forward and backward time evolution known as out-of-time-order correlators are commonly used for diagnosing quantum chaos in many-body systems~\cite{maldacena2016bound} with applications ranging from superconductivity~\cite{larkin1969quasiclassical} to black hole physics~\cite{shenker2014multiple,Kitaev2014talk}.

The ability of conventional classical computers to simulate Hamiltonian dynamics is  limited by the exponential cost
of representing  entangled quantum states.
Consider the simplest version of the problem -
simulating dynamics  of a one-dimensional spin chain with short-range interactions starting from an initial unentangled state.
The entanglement entropy between the left and the right halves of the chain 
typically grows linearly with the evolution time $t$ until it saturates at the value proportional to the system size $n$,
see for instance~\cite{calabrese2005evolution}.
Most classical simulation methods  rely on Matrix Product States (MPS) to approximate
the time-evolved states. Since the bond dimension of MPS grows exponentially with the amount of entanglement,
the simulation cost grows exponentially with $n$ or $t$.
This limits applicability of classical simulators to small systems or short evolution times.

In contrast, quantum computers can efficiently simulate Hamiltonian dynamics for most practically relevant Hamiltonians~\cite{lloyd1996universal}, at least in theory. 
For example, quantum algorithms based on high-order Trotter formulas~\cite{childs2021theory} or the Lieb-Robinson bound~\cite{haah2021quantum}
can simulate dynamics of spin chain Hamiltonians with 
gate complexity scaling almost linearly with space-time volume $nt$.
It is expected that simulation problems of this type with space-time volume $nt\approx 10^4$ are already intractable for existing classical computers~\cite{childs2018toward}.

It should be emphasized that even quantum computers cannot simulate Hamiltonian dynamics exactly.
To begin with, the existing quantum processors are not fault-tolerant and the accumulation of errors limits the depth of quantum
circuits that can be executed reliably. Even leaving the fault-tolerance problem aside, a quantum computer can only
approximate the exact time evolution unitary due to algorithmic errors known as Trotter errors~\cite{childs2021theory}.
Although the asymptotic runtime of best known quantum algorithms scales only logarithmically with the desired error tolerance~\cite{low2017optimal},
the cost of high-precision simulations is prohibitive for near-term quantum processors.
For example, the benchmark problem
of Ref.~\cite{childs2018toward}
with space-time volume $nt\approx 10^4$ would require nearly $10^7$ CNOT gates to approximate the time evolution
within three digits of precision~\cite{childs2019faster,bravyi2022future}.
This exceeds the size of quantum circuits demonstrated to date by many orders of magnitude. 

A natural question is whether classical and quantum simulation algorithms 
working in tandem  can accomplish Hamiltonian dynamics simulation at a lower cost compared with the
classical or quantum algorithms alone.
% mention circuit knitting method peng2020simulating ?
Here we begin addressing this question by showing how to combine classical tensor network algorithms for simulating weakly entangling
quantum circuits~\cite{vidal2003efficient} with the quantum simulation algorithms based on Multi Product Formulas~\cite{childs2012hamiltonian,low2019well,vazquez2023well,rendon2022improved,zhuk2023trotter}. 

To illustrate the key ideas of our approach, let us begin with a simpler problem: estimating the Trotter error.
Suppose $H$ is a Hamiltonian describing a spin chain with short-range interactions and $|\psi_0\ra$
is an initial weakly entangled state. Evolving the spin chain 
for a time $t$
results in the state $e^{-itH}|\psi_0\ra$. A Trotter circuit approximating $e^{-itH}$ has the form
$U(t)=S(t/k)^k$, where $k$ is the number of Trotter steps
and $S(\tau)$ is a shallow quantum circuit  approximating $e^{-i\tau H}$ for 
a small evolution time $\tau$. We choose $S(\tau)$ as the second-order Trotter-Suzuki formula~\cite{childs2021theory}.
The fidelity between the exact time-evolved state and its Trotter approximation
is controlled by the overlap $L(t) = \langle \psi_0|U(t)^\dag e^{-itH}|\psi_0\rangle$. 
As detailed below,
 one can view the operator $U(t)^\dag e^{-itH}$ as a combination of a forward
time evolution generated by $H$ and a backward time evolution generated by a time-dependent Hamiltonian
associated with $U(t)$. In the regime when the Trotter error is small, the forward and the backward time evolutions
nearly cancel each other. Thus one might expect that the unitary $U(t)^\dag e^{-itH}$ is weakly entangling, even though
the time-evolved state $e^{-itH}|\psi_0\ra$ may have a lot of entanglement. 
If this is the case, one can use a generalization of MPS algorithms based on Matrix Product Operators (MPO)
to approximate the unitary $U(t)^\dag e^{-itH}$ on a classical computer~\cite{schollwock2011density}. 
Computing the overlap $L(t)$ amounts to contracting the MPO representation of  $U(t)^\dag e^{-itH}$ with an MPS representation
of  $|\psi_0\ra$. 

Once the Trotter error can be reliably estimated, the next step is to reduce this error.
Our approach is based on dynamic Multi Product Formulas (MPF) 
proposed in~\cite{zhuk2023trotter}. Instead of approximating the Hamiltonian dynamics
by a single Trotter circuit, an MPF combines several Trotter circuits $U_a(t)=S(t/k_a)^{k_a}$
with a varying number of Trotter steps $k_a$. An MPF with $r$ terms 
is defined as a linear combination of density matrices associated with the states $U_a(t)|\psi_0\ra$
for $a=1,\ldots,r$.
Coefficients in this linear combination serve as variational parameters that we optimize to reduce
the Trotter error. Specifically, we minimize the Frobenius norm distance between the
density matrix of the
exact time-evolved state and the approximating MPF (the latter may or may not be a valid density matrix).
This distance is shown to be a simple function 
of the MPF coefficients as well as overlaps $\la \psi_0|U_a(t)^\dag e^{-itH}|\psi_0\ra$ 
and $\la \psi_0|U_a^\dag(t) U_b(t)|\psi_0\ra$. We show how to compute these overlaps using the MPO
algorithm sketched above. Minimizing the distance amounts to solving a simple least-squares problem.
Crucially, a quantum processor only needs to prepare and measure
 individual states  $U_a(t)|\psi_0\ra$ while different terms in the MPF are combined by a classical
 post-processing of the measured results.

\begin{figure*}
\centering
\includegraphics[width=\textwidth]{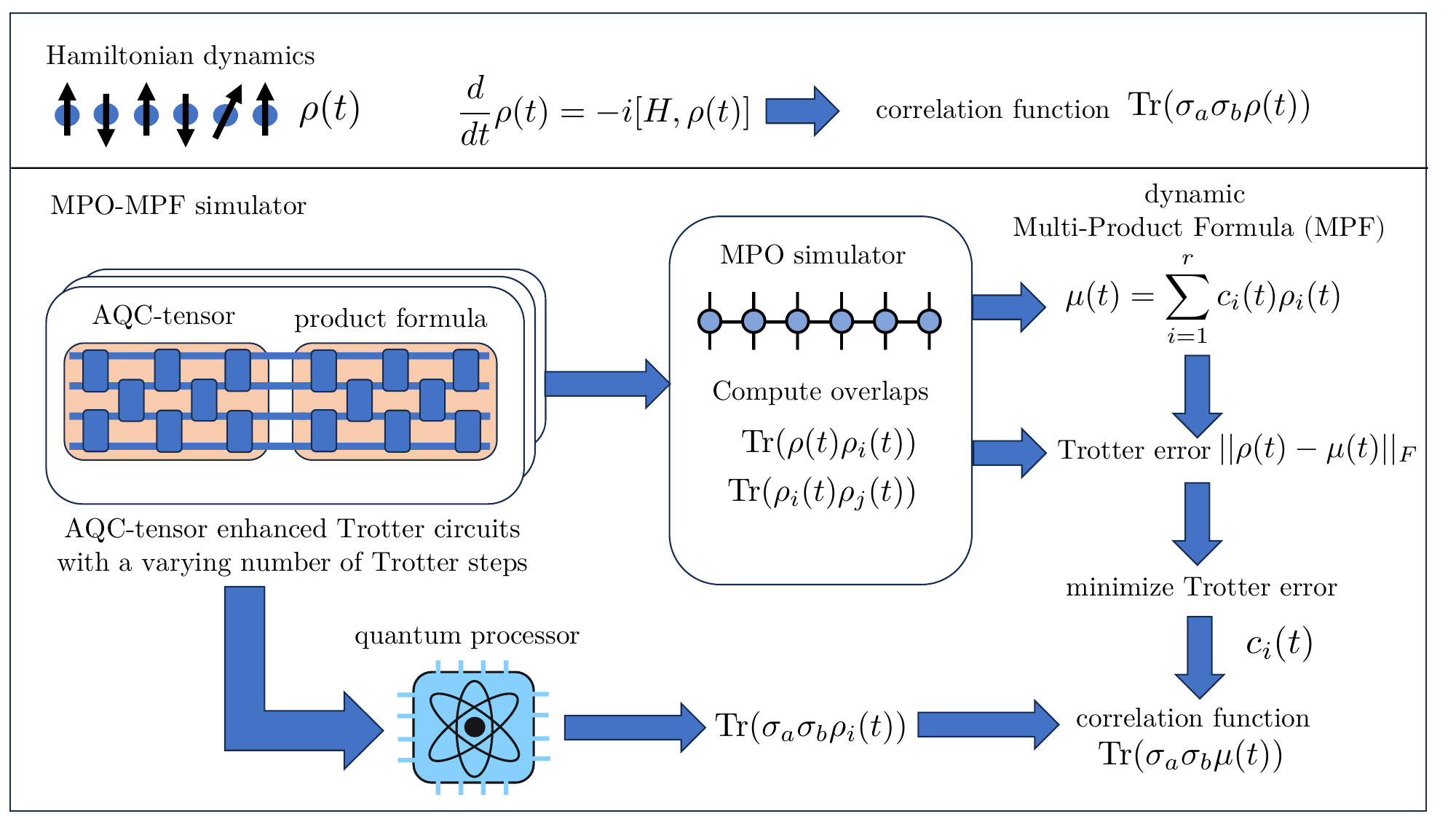}
\caption{Top panel: Hamiltonian dynamics for a spin chain Hamiltonian $H$. The goal is to compute time-dependent correlation functions such as the expected value of Pauli $\sigma_a \sigma_b$
on the time-evolved state $\rho(t)$. Bottom panel: simulation workflow.
The time-evolved state $\rho(t)$ is approximated by a linear combination of states $\rho_i(t)$ generated by 
an ensemble of Trotter circuits with a varying number of Trotter steps. The initial segment of each Trotter circuit
is compressed using the $AQCtensor$ algorithm of Ref.~\cite{robertson2023approximate}.
A classical simulator based on Matrix Product Operators (MPO) is responsible for computing overlaps for the chosen ensemble of states.
Our approximation for $\rho(t)$ is a dynamic Multi Product Formula (MPF)  $\mu(t)=\sum_{i=1}^r c_i(t) \rho_i(t)$
with coefficients $c_i(t)$ optimized to minimize the Trotter error measured by the Frobenius norm distance
between $\rho(t)$ and $\mu(t)$.  This distance is  a simple function of the overlaps
$\mathrm{Tr}(\rho(t) \rho_i(t))$ and $\mathrm{Tr}(\rho_i(t) \rho_j(t))$. 
A quantum processor is responsible for
estimating correlation functions associated with each individual
state $\rho_i(t)$. Finally, different terms in the MPF are combined 
to obtain an estimate $\mathrm{Tr}(\rho(t) \sigma_a \sigma_b) \approx  \sum_{i=1}^r c_i(t) \mathrm{Tr}(\rho_i(t) \sigma_a \sigma_b)$.
The classical cost of computing overlaps such as $\mathrm{Tr}(\rho(t) \rho_i(t))$
is small due to an approximate  cancellation between forward and backward time evolution
generated by $H$ and by the Trotter circuit.
}\label{fig:workflow}
\end{figure*}

The final ingredient of our simulation algorithm is $AQCtensor$ — an approximate quantum compiler based on tensor networks proposed in~\cite{robertson2023approximate} which aims to find a short depth quantum ciruit that closely approximates a given Matrix Product State. In the workflow proposed here, $AQCtensor$ is used to prepare a state $\ket{\psi_{t_1}}=e^{-iHt_1}\ket{\psi_0}$ where $\ket{\psi_0}$ is some initial product state and where $t_1$ is the largest time at which the time-evolved quantum state can be stored classically as an MPS with a given desired accuracy and fixed amount of memory. Specifically, the $AQCtensor$ algorithm optimizes a parameterized circuit $V(\Theta)$ such that $V(\Theta)\ket{\psi_0}\approx e^{-iHt_1}\ket{\psi_0}$. Additional Trotter steps are then appended to this parametric circuits to simulate the evolution for a time given by $t_2$, so that the total time simulated by the full $AQCtensor$ + Trotter circuit is $t=t_1+t_2$. The full simulation workflow combining $AQCtensor$ and our \emph{MPO}-based dynamic multiproduct algorithm is illustrated in Figure~\ref{fig:workflow}.

To put our results in a broader context, let us discuss previously known simulation algorithms combining quantum and classical workloads.
Hybrid tensor network  studied in~\cite{yuan2021quantum,haghshenas2022variational,eddins2022doubling,schuhmacher2024hybrid} aim to reduce the cost of storing high-rank tensors in a classical memory.
This is achieved by augmenting a classical efficiently contractable tensor network such as MPS with quantum tensors
whose components are defined as amplitudes of easy-to-prepare quantum states. 
It is envisioned that hybrid tensor networks are applicable to both static (ground state) and dynamic simulation problems~\cite{yuan2021quantum}.
Our MPO-MPF simulator is closely related to  the tensor network error mitigation method proposed in~\cite{filippov2023scalable}. The latter aims to reduce
errors due to qubit decoherence  and imperfect control, as opposed to algorithmic errors considered in the present work. 
In particular, our MPO algorithm for estimating Trotter error is similar to the middle-out contraction method used in~\cite{filippov2023scalable}.

This work is structured as follows: in section \ref{sec:background} we recall the basic definitions used in previous work on multiproduct formulas and $AQCtensor$. In section \ref{sec:mpo-mpf} we introduce our new MPO-based dynamic MPF algorithm. We discuss the errors involved and show in Figure \ref{bootstrap} how our \emph{MPO-MPF algorithm} provides an advantage over purely classical or purely quantum approaches in isolation - we leave a more detailed error analysis to Appendix \ref{sec:error-analysis}. In section \ref{sec:results} we present the results of both our classical and quantum simulations - see Figures \ref{fig:heron_expec} and \ref{fig:AQC_observables}. We also introduce a classically implementable numerical test to determine the simulation times at which any given multiproduct formula will provide an advantage over a single product formula - see the inequality in (\ref{r_plus_1_test}). We stress that this numerical test applies both to the dynamic MPF algorithm proposed here and the previously proposed static MPF algorithms \cite{childs2012hamiltonian, low2019well, vazquez2023well}. We conclude in section \ref{sec:discussion}.

\section{Background}\label{sec:background}
\subsection{Multiproduct Formulas}
In what follows we briefly recall from~\cite{zhuk2023trotter} the notions of Trotter product and multi-product formulas for quantum time evolution. A product formula is a quantum circuit $S(t)$ approximating
the evolution operator $e^{-itH}$ of a quantum system with a Hamiltonian $H$: a common approach is to split the time window $(0,t)$ into $k$ intervals of length $t/k$, construct a circuit $S(t/k)$ and apply it $k$
times. The depth of the resulting circuit $S(t/k)^k$ scales linearly with $k$, and $k$ depends on the structure of $H$, desired approximation error, and the type of product formula~\cite{zhuk2023trotter}. %For example, solving a benchmark simulation problem posed by Childs, Maslov, et al.~\cite{childs2018toward} for a system of $100$ qubits would require $k\ge 5000$ time steps if one uses the fourth-order Trotter product formula.

We stress that Trotter product formulas that accurately approximate Hamiltonian dynamics might be very deep. Multi-product Formulas (MPF)~\cite{childs2012hamiltonian} allow for the reduction of the depth of a circuit without increasing its approximation error. The latter point can be illustrated by the following example: for a given sequence of integers $k_1,\ldots,k_r$, consider a \emph{static MPF} in the form of a real linear combination $\mu^S(t)=\sum_{j=1}^r c_j \rho_{k_j}(t)$ of $r$ density matrices $\rho_{k_j}(t)=S(t/{k_j})^{k_j}\rho_{0}S(t/{k_j})^{-k_j}$ each of which approximates the exact time-evolved density matrix $\rho(t)=e^{-itH} \rho_{0} e^{itH}$ where $\rho_{0} =  |\psi_{0}\rangle\langle \psi_{0}|$ and $S$ is a Trotter product formula of order $p$. The coefficients $c_j$ are chosen to solve the following linear system:
\begin{equation}\label{eq:MPFlinsys}
    \sum\limits_{i=1}^r c_i = 1 \ \text{and} \sum\limits_{i=1}^r \frac{c_i}{k_i^q} = 0
  \end{equation}
  for $q \in \{p, p+1, ..., 2p-1\}$. This choice of coefficients guarantees that the resulting MPF $\mu^S$ has at least quadratically (or even exponentially) better accuracy (measured in $1$-norm, $\|\rho-\mu^S\|_1$) than each $\rho_{k_j}$ without using circuits that are any deeper than those used to produce $\rho_{k_j}$ \cite{zhuk2023trotter}. As a result, the expected value
of any observable $\Tr(\calO \rho(t))$ can be approximated better by a linear combination $\sum_{j=1}^r c_j  \Tr(\calO \rho_{k_j}(t))$.

By design, the static MPF $\mu^S(t)$ reduces the Trotter error: $\|\rho-\mu^S\|_1< \|\rho-\rho_{k_j}\|_1$. However this error is not minimized: i.e. $\min_{c_i}\|\rho-\sum_{j=1}^r c_j \rho_{k_j}\|_1<\|\rho-\mu^S\|_1$. In addition, the coefficients $c_i$ found from~\eqref{eq:MPFlinsys} are sensitive to the choice of $k_i$ as for certain $k_i$ the set of linear equations in~\eqref{eq:MPFlinsys} is ill-conditioned numerically. In this scenario, the resulting vector of coefficients has a large condition number $\sum_i |c_i|$ which leads to an amplification of the sampling noise for expectation values evaluated on a quantum computer. These two points were addressed in~\cite{zhuk2023trotter} by introducing dynamic multiproduct formulas: $\mu^D(t)=\sum_{j=1}^r c_j(t) \rho_{k_j}(t)$
with time-dependent coefficients $c_j(t)$ chosen to minimize the error $\|\mu^D(t)-\rho(t)\|_F$ measured in Frobenius norm. It can be computed by solving a convex optimization problem in which a cost function $E_F^D$ is minimized with respect to coefficients $c_j(t)$:
\begin{align}
\label{proj_error}
E_F^D = \| \rho(t)-\mu^D(t)\|_F^2 &=  1 + \sum_{i,j=1}^r M_{ij}(t) c_i(t) c_j(t) \nonumber \\
&- 2\sum_{i=1}^r L_i(t) c_i(t)
\end{align}
where $M(t)$ is the Gram matrix with elements 
\begin{align}\label{eq:Mt}
M_{ij}(t) & = \Tr(\rho_{k_i}(t) \rho_{k_j}(t)) \nonumber \\
& = \left| \langle \psi_{0} | S(t/k_i)^{-k_i} S(t/k_j)^{k_j}|\psi_{0}\rangle\right|^2
\end{align}
and $L(t)$ is a vector of overlaps with components 
\begin{equation}\label{eq:Lexact}
L_j(t) = \Tr(\rho(t) \rho_{k_j}(t)) =  \left| \langle \psi_{0} | S(t/k_j)^{-k_j} e^{-i H t}|\psi_{0}\rangle\right|^2
\end{equation}
We discuss in detail the relationship between the $1$-norm and the Frobenius norm in section \ref{sec:results}. In particular, we discuss the conditions under which dynamic MPF outperforms each Trotter product formula as measured in $1$-norm error - see equation (\ref{r_plus_1_test}) and surrounding text. It was noted in~\cite[Section 5.a)]{zhuk2023trotter} that dynamic MPF outperforms the static one at long simulation times, and it does not require that the $k_i$ be constrained in order to get well-conditioned coefficients. %so even for shorter times if    where we note the following relationship between the Fronbenius norm error and the one-norm error: $\|\rho-\mu^D\|_F\le\|\rho-\mu^D\|_1\le\sqrt{r+1}\|\rho-\mu^D\|_F$. 
%In other words, $\mu^D(t)$ is the optimal projection (in Frobenius norm) of $\rho$ onto linear  subspace generated by $\rho_{k_j}$.
%\textbf{TODO:SZ-done} Make notation consistent with the rest of the paper, i.e. $\ket{\psi_{0}} \rightarrow \ket{psi_0}$

%%%%%%%%%%%%%%%%%%%%%%%%%%%%%%%%%%%%%%%%%%%%%%%%%%%%%%%%%%%%%%%%%%%%%%%%%%%%%%%%%%%%%%%%%%%
\subsection{$AQCtensor$ algorithm}\label{sec:aqctensor}
%%%%%%%%%%%%%%%%%%%%%%%%%%%%%%%%%%%%%%%%%%%%%%%%%%%%%%%%%%%%%%%%%%%%%%%%%%%%%%%%%%%%%%%%%%%
The $AQCtensor$ algorithm was proposed in \cite{robertson2023approximate} and we briefly recall it here. The algorithm aims to find a short depth quantum circuit approximating a given Matrix Product State. It is based on classical optimization; one starts with a target quantum state $\ket{\psi_t}$ and a variational circuit Ansatz $V(\Theta)\ket{\psi_0}$ with $\ket{\psi_0}$ some initial product state. Both the target state and Ansatz state are stored as Matrix Product States. The cost function:
\begin{equation}\label{eq:cost-func}
    C(\Theta) = 1 - \left|\bra{\psi_0}V^{\dagger}(\Theta)\ket{\psi_t}\right|^2 
\end{equation}
is evaluated by computing the overlaps of the MPS $\ket{\psi_t}$ and $V(\Theta)\ket{\psi_0}$ and is minimized with respect to rotation angles $\Theta$ resulting in a short depth circuit $V(\Theta)\ket{\psi_0}$ representing the target state. In this work, we employ $AQCtensor$ to get a short-depth representation of the target state $\ket{\psi_{t_1}}=e^{-iHt_1}\ket{\psi_0}$. The idea is to generate $\ket{\psi_{t_1}}$ using a classical algorithm such as TEBD for the latest possible time $t_1$ that can be stored for a given maximum bond dimension $\chi_{max}$ and for a given desired precision determined by the bond dimension truncation threshold which we denote by $\lambda_0$. Instead of using standard Trotterization to prepare the state $\ket{\psi_{t_1}}$ (which may require a deep circuit), one uses the optimized angles $\Theta_{opt}$ (obtained by minimizing the cost function) to prepare $V(\Theta_{opt})\ket{\psi_0}$. Finally, one applies a Trotter circuit to this state to obtain a time-evolved state beyond the time at which it could be stored classically: $S(\frac{t_2}{k})^kV(\Theta_{opt})\ket{\psi_0}$ such that the total evolution time simulated by the full circuit is $t=t_1+t_2$. As discussed in \cite{robertson2023approximate}, we use a circuit Ansatz $V(\Theta)$ which has the same CNOT structure as a second-order Trotter circuit allowing for the use of a ``smart-initialization" scheme where the initial values $\Theta_0$ are set such that $V(\Theta_0)$ corresponds exactly to a second-order Trotter circuit. Therefore, after one step of a minimization algorithm such as gradient descent it is assured that the optimized circuit Ansatz will more closely represent the exact time-evolved circuit than a standard Trotter circuit.

\section{MPO-based Dynamic Multiproduct Formulas}\label{sec:mpo-mpf}
Here we present a Matrix Product Operators (MPO) based method to calculate the dynamic coefficients $c_j(t)$ classically, leaving the calculation of the expectation values $\Tr(\mathcal{O}\rho_{k_j})$ to the quantum computer. We argue that, in the context of dynamic MPF, the bond dimension of MPO-based method scales (in time) more favourably than that of Matrix Product States (MPS), and we also show that the error (in computing expectation values) due to MPO bond dimension truncation can be made negligible compared to the error of the individual Trotter product formulas used in the multiproduct. 
%discuss the impact of the error of MPO representation on the accuracy of the MPF. 

%\paragraph{MPO-algorithm for computing overlaps.} 
Recall from~\eqref{eq:Mt}-\eqref{eq:Lexact} that $M_{ij}$ is the overlap of the quantum states $\ket{\psi_j}=S(t/k_j)^{k_j}\ket{\psi_0}$ and $\ket{\psi_i}=S(t/k_i)^{k_i}\ket{\psi_0}$ and $L_j$ is the overlap of the quantum states $\ket{\psi_j}$ and $\ket{\psi_{ex}}=e^{-i H t}\ket{\psi_0}$. For large times and large systems it is not classically efficient to store $\ket{\psi_i}, \ket{\psi_j}, \ket{\psi_{ex}}$ as Matrix Product States, even for 1D systems. To calculate $M_{ij}$ and $L_j$ and hence $c_{j}(t)$, we must therefore avoid explicitly storing the quantum states. To do so, we define the objects $F_{ij}$ and $F_{ex,j}$:
\begin{equation}\label{eq:Fdefs}
\begin{aligned}
    F_{ij} &\equiv S\left(\frac{t}{k_i}\right)^{-k_i}S\left(\frac{t}{k_j}\right)^{k_j}\\
    F_{\rm ex, j} & \equiv e^{iHt}S\left(\frac{t}{k_j}\right)^{k_j}
\end{aligned}
\end{equation}
The quantities $M_{ij}$ and $L_j$ are thus given by:
\begin{equation}\label{mij}
   M_{ij} =  |\bra{\psi_0}F_{ij}\ket{\psi_0}|^2
\end{equation}
and
\begin{equation}\label{lj}
   L_{j} =  |\bra{\psi_0}F_{\rm ex, j}\ket{\psi_0}|^2
\end{equation}
In practice, we approximate $e^{-iHt} \approx S\left(\frac{t}{k_0}\right)^{k_0}$ with $k_0 \gg k_i$. In what follows, we thus consider $F_{\rm ex, j}$ to be a special case of $F_{ij}$ with $i=0$. In Algorithm \ref{alg:MPO} we present our algorithm to calculate $F_{ij}$.
\begin{algorithm}[H]
    \caption{MPO-based algorithm to calculate $M_{ij}$ and $L_j$}\label{alg:MPO}
    \begin{algorithmic}
        \State $F_{ij} \gets \mathds{1}$
        \While {evolved$\_$time$\_$i $<$ t and evolved$\_$time$\_$j $<$ t}
            \If{evolved$\_$time$\_$j $\leq$ evolved$\_$time$\_$i}
                \State $F_{ij} \gets F_{ij} S\left(\frac{t}{k_j}\right)$
            \Else
                \State $F_{ij} \gets S^{\dagger}\left(\frac{t}{k_i}\right)F_{ij}$
            \EndIf
        \EndWhile
    \end{algorithmic}
\end{algorithm}

\paragraph{Memory efficiency.} We now briefly comment on the memory requirements in Algorithm \ref{alg:MPO} while a more detailed analysis can be found in appendix~\ref{sec:error-analysis}. More specifically, we compare the bond dimension of the MPO required to store $F_{ij}$ with a given precision vs the bond dimension of the MPS required to store the time evolved state $e^{-iHt}\ket{\psi_0}$. For a general Hamiltonian (i.e. one that has not been handpicked to exhibit behaviour such as many-body localization), the bond dimension required to store $e^{-iHt}\ket{\psi_0}$ scales exponentially in time, i.e. $\log(\chi_{mps}) \propto T$.  However, unlike the MPS for the quantum state, the bond dimension required to store $F_{ij}$ and $F_{ex,j}$ as Matrix Product Operators decreases with decreasing Trotter step $dt$. In appendix \ref{sec:error-analysis}, we argue that $\log(\chi_{mpo}) \propto Tdt^2$ in the worst case scenario. For any given simulation time $t$, we thus argue that there exists a Trotter step $dt$ such that it is more efficient to store the MPOs $F_{ij}$ than to store the MPS for the quantum states individually.

\paragraph{Error analysis.} Now consider the effect of truncating the bond dimension of the MPOs (for computing overlaps $M_{ij}$ and $L_j$) on the accuracy of estimating expectation values. In Appendix \ref{sec:error-analysis}, we argue that for any given truncation threshold $\lambda_0$ on the bond dimension, there exists a value of the Trotter step number $k$ such that: 
\begin{itemize}
    \item [1)] the memory resources required to store $F_{ij}$ are lower than those required to store the quantum state $e^{-iHt}\ket{\psi_0}$ using MPS,
    \item[2)] for an observable $\calO$, MPO truncation error $\varepsilon(\lambda_0)$, Trotter error $\calE$, and coefficients $c_j$ obtained using Algorithm~\ref{alg:MPO} one has $
    |\Tr(\calO \rho(t))-\sum_{j=1}^r c_j(t) \Tr(\calO\rho_{k_j}(t)))|\le |\Tr(\calO (\rho(t)-\mu^D(t)))|\\ + O(\varepsilon(\lambda_0)\calE)$
\end{itemize}
Point $2)$ suggests that the errors on the observables arising from the finite bond dimension are dominated by $|\Tr(\calO (\rho(t)-\mu^D(t)))|$ -- the error of the dynamic MPF with exact $M_{ij}$ and $L_j$, while the error due to MPO truncation, $O(\varepsilon(\lambda_0)\calE)$ is proportional to the product of the Trotter error and truncation error, i.e. there are no first order error terms arising from the Tensor Network part of our algorithm - see equation (\ref{eq:mpf-mpo-error-obs}) and surrounding text.

Below we further illustrate point $1)$ above numerically for the well-studied Heisenberg model with Hamiltonian $H$ given by:
\begin{equation}\label{heisenberg}
    H = -\sum\limits_{i=1}^{L-1} (S^x_i S^x_{i+1} + S^y_i S^y_{i+1} + S^z_i S^z_{i+1}) 
\end{equation}
with $S^x_i = \frac{1}{2}\sigma^x_i$, $S^y_i = \frac{1}{2}\sigma^y_i$ and $S^z_i = \frac{1}{2}\sigma^z_i$ where $\sigma^x$, $\sigma^y$ and $\sigma^z$ are the Pauli matrices. We consider the initial state to be the N{\'e}el state: $\ket{\psi_0} = \ket{1010...}$. In Figure \ref{bootstrap} we compare three different simulation methods: 
\begin{itemize}
    \item $\rho_{k_{max}}=\ket{\psi_{k_{max}}}\bra{\psi_{k_{max}}}$: $k_{max}=\max\{k_1,k_2\}$, $\ket{\psi_{k_{max}}}$ is 2nd-order Trotter product formula implemented using MPS-based classical simulation \cite{tenpy} with large bond dimension ($\chi = 400$). % which can be seen as mimicing the behaviour of a Trotter circuit on a quantum computer,
    \item $\rho_{dynMPF}=c_1 \ket{\psi_{k_{1}}}\bra{\psi_{k_{1}}}+c_2\ket{\psi_{k_2}}\bra{\psi_{k_{2}}}$: $c_1$, $c_2$ computed by minimizing~(\ref{proj_error}) provided the overlaps $M_{ij}$ and $L_j$ are computed as per Algorithm \ref{alg:MPO} with low MPO bond dimension ($\chi = 50$), $\ket{\psi_{k_{1}}}$ is computed similarly to $\ket{\psi_{k_{max}}}$.
    \item $\rho_{MPS}=\ket{\psi_{MPS}(t)}\bra{\psi_{MPS}(t)}$: $\ket{\psi_{MPS}(t)}$ calculated using MPS-based classical simulation \cite{tenpy} with relatively low MPS bond dimension ($\chi=100$) and a very low time step $dt$ such that there is effectively no Trotterization error. 
\end{itemize}
Each of the above methods is then characterized by its error defined as the distance (in Frobenious norm) to a ``quasi-exact" density matrix $\rho$. More specifically, the error of $\rho_{k_{max}}$ is given by
\begin{equation}\label{efkmax}
\begin{aligned}
    E_F^{k_{max}} & = \|\rho-\rho_{k_{max}}\|_F ^2\\
    & = 2 - 2 \Tr(\rho \rho_{k_{max}}) \\
    & = 2 - 2 |\bra{\psi_0}S^{k_{max}}\left(\frac{t}{k_{max}}\right)e^{-iHt} \ket{\psi_0}|^2 
\end{aligned}
\end{equation}
Here $e^{-iHt} \ket{\psi_0}$ is computed with 4th-order Trotter product formula, a very low time step $dt$ and a large bond dimension ($\chi = 400$).  Similarly, the Frobenius norm error of the MPS simulation is given by 
\begin{equation}\label{efmps}
    E_F^{MPS} = \|\rho-\rho_{MPS}\|_F ^2 = 2 - 2| \bra{\psi_{MPS}(t)} e^{-iHt} \ket{\psi_0} |^2
\end{equation}
where $\ket{\psi_{MPS}(t)}$ is an approximation to the exact state after truncation of the bond dimension and renormalization of the singular values such that $|\braket{\psi_{MPS}(t) | \psi_{MPS}(t)}|^2 = 1$. The Frobenius norm of the dynamic MPF error is defined in equation (\ref{proj_error}) - however for a fair comparison with $E_F^{MPS}$ and $E_F^{k_{max}}$ one must take great care with how (\ref{proj_error}) is applied. In particular, we first calculate $M_{ij}$ and $L_j$ using Algorithm \ref{alg:MPO} and denote the resulting quantities by $M_{ij}^{MPO}$ and $L_j^{MPO}$. We then input these quantities to the RHS of equation (\ref{proj_error}) and find the coefficients $c_i^{MPO}$ that minimise this quadratic function, subject to the constraint $\sum\limits_i c_i = 1$. We then calculate $M_{ij}$ and $L_j$ quasi-exactly, i.e. using heavy numerical simulations with bond dimension $\chi = 400$. We denote the resulting quantities by $M_{ij}^{ex}$ and $L_j^{ex}$. We then define the quantity $E_F^{MPO-MPF}$: 
\begin{equation}\label{Efmpo}
\begin{aligned}
E_F^{MPO-MPF} =  1 + &\sum_{i,j=1}^r M_{ij}^{ex}(t) c_i^{MPO}(t) c_j^{MPO}(t)  \\
&- 2\sum_{i=1}^r L_i^{ex}(t) c_i^{MPO}(t)
\end{aligned}
\end{equation}
In Figure \ref{bootstrap}, we see that $E_F^{MPO-MPF}$ has the lowest error of the three methods, despite having the lowest bond dimension in its classical part of the workflow (i.e. for the calculation of the coefficients $c_i^{MPO}(t)$).\\

\begin{figure}
\centering
\includegraphics[width=\columnwidth]{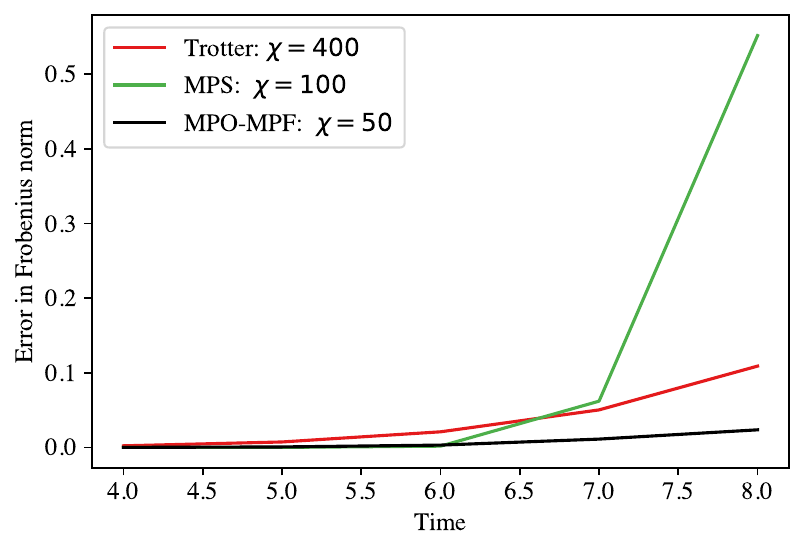}
\caption{The Frobenius norm of the error in the density matrices produced by three different simulation methods of the time evolution with the $50$-qubit Hamiltonian in (\ref{heisenberg}) acting on the initial state $\ket{\psi_0}=\ket{1010...}$. Red line: $E_F^{k_{max}}$ in equation (\ref{efkmax}) with $k_{max}=k_1=18$, $\chi=400$. Green line: $E_F^{MPS}$ in equation (\ref{efmps}) with $\chi=100$ and very small time step ($dt=0.025$). Black line: $E_F^{MPO-MPF}$ in equation (\ref{Efmpo}) with $\chi=50$ and two Trotter formulas used in the multiproduct, $k_1=18$ and $k_2=14$. The error arising from dynamic MPO-MPF with bond dimension $\chi=50$ is lower than pure Trotter with bond dimension $\chi=400$ (i.e. a proxy for a quantum device) and for a pure MPS state simulation with bond dimension $\chi=100$. The total memory requirements to store the MPS on $L$ qubits, i.e. the total number of floating point complex numbers, is given by $2L\chi_{mps}^2$ while for the MPO this number is $4L\chi_{mpo}^2$. The total classical memory requirements for the $\chi=100$ MPS curve are thus larger than those of the $\chi=50$ MPO curve. When implemented on a real quantum device, the shot noise is magnified by a factor given by the one-norm of the coefficients $\sum\limits_i|c_i(t)|$. One must thus ensure that a given multiproduct formula is well-conditioned - see e.g. \cite{vazquez2023well}. The multiproduct formula considered here is indeed well conditioned - the one-norm $\sum\limits_i|c_i(t)|$ produced from the MPF algorithm is less than $4.5$ for all times considered.}\label{bootstrap}
\end{figure}

\section{Results}\label{sec:results}
In this section we provide a thorough numerical evaluation of the quantum/classical workflow outlined in Section~\ref{sec:mpo-mpf}. In particular, we calculate dynamic MPF coefficients $c_j(t)$ classically using Algorithm \ref{alg:MPO} and we use a quantum computer to calculate expectation values $\Tr(\calO \rho_{k_j})$ for a local observable $\calO$, namely for one and two-site operators $\sigma^z_i$ and $\sigma^z_i\sigma^z_{i+1}$, and for a Hamiltonian similar to~\eqref{heisenberg}, but with adjusted coefficients to ensure that it is not integrable:
\begin{equation}\label{h_general}
    H = -\sum\limits_{i=1}^{L-1} (J_i(S^x_i S^x_{i+1} + S^y_i S^y_{i+1}) + \Delta_iS^z_i S^z_{i+1} )
\end{equation}
In~\eqref{h_general} $J_i$ are sampled from a uniform distribution supported on $[\frac{1}{4},\frac{3}{4}]$, and $\Delta_i = 2J_i$. In what follows we show that dynamic MPF with ``shallow product formulas" can provide a comparable precision to that of a deeper (and hence more accurate) product formula in the case of classical and more importantly quantum simulations thus allowing to operate shallow circuits without loss of precision. 
\paragraph{Evaluation on classical simulations.} Our full workflow is as follows: first we determine the number of Trotter steps $k$ that are required to keep the Trotter error below a given target precision for a second-order Trotter formula, $S_2(t)$. In principle one can obtain an estimate for the number of Trotter steps required by using the rigorous upper bound on the error~\cite{childs2021theory}, however this upper bound accounts for a worst-case scenario, and so it is very likely to significantly overestimate the required number of Trotter steps in practice. Instead, we pick a large $k$ and use Algorithm \ref{alg:MPO} to generate $F_{ex,k}$ in equation (\ref{eq:Fdefs}) and then use equation (\ref{efkmax}) to compute an estimate of the Trotter error using:
\begin{equation}\label{efkF}
    E_F^k = 2-2|\bra{\psi_0} F_{\rm ex,j} \ket{\psi_0}|^2
\end{equation}
Then we pick a set of Trotter steps $k_j$ such that $k_j < k$ for all $j$ and we use equation (\ref{proj_error}) to obtain an estimate of the dynamic MPF error $E_F^D$ arising from these $k_j$. The quantities $M_{ij}$ and $L_i$ in (\ref{proj_error}) are obtained by calculating $F_{ij}$ and $F_{ex,j}$ using Algorithm \ref{alg:MPO} followed by the application of equations (\ref{mij}) and (\ref{lj}). We choose the minimum values of $k_j$ such that the resulting dynamic MPF, $\mu^D(t)=\sum_{j=1}^r c_j(t) \rho_{k_j}(t)$ of $r$ density matrices $\rho_{k_j}(t)=S_2(t/{k_j})^{k_j}\rho_{0}S_2(t/{k_j})^{-k_j}$ has a smaller algorithmic error than the deeper single Trotter circuit, i.e.  $E_F^D(k_j) \leq E_F^k$ with $k>k_j$ for all $j$. In the top panel of Figure \ref{trott_errors_mpo_mps}, we plot the quantity $E_F^D$ - i.e. the error in Frobenius norm for $\mu^D$ with three terms, $k_1=2$, $k_2=3$ and $k_3=4$. We compare this with the Frobenius norm error $E_F^k$ for a single Trotter circuit with $k=6$ and also with a single Trotter circuit with $k=4$. We find that the dynamic MPF with three relatively shallow circuits has a comparable algorithmic error to the deeper Trotter circuit up until $t \approx 4.1$ - this crossover point is represented by the dotted blue vertical line labelled ``Trotter test". We plot the same quantities in the bottom panel of Figure \ref{trott_errors_mpo_mps} but where $E_F^D$ and $E_F^k$ have been calculated by explicitly storing the states $e^{-iHt}\ket{\psi_0}$ and $S\left(\frac{t}{k_j} \right)^{k_j}\ket{\psi_0}$ as Matrix Product States and calculating their overlaps. Clearly this is not a scalable approach as the bond dimension required to do this increases exponentially with simulation time, but we do so here for demonstration purposes and take bond dimension $\chi=400$. We observe that the (numerically demanding) MPS simulation and the (much less numerically demanding) MPO simulation in Figures \ref{trott_errors_mpo_mps} predict the same crossover point between the errors of the $k=6$ Trotter circuit and the dynamic MPF circuits, demonstrating the resilience of the ``Trotter test" to the truncation error induced by the low bond dimension used in the MPO algorithm. Furthermore, this result demonstrates the scalability of Algorithm~\ref{alg:MPO} to larger time scales at which it would no longer be possible to store the quantum states accurately as Matrix Product States. Note that one can also use this MPO-based algorithm to test the validity of implementing a \textit{static} multiproduct formula \cite{vazquez2023well, zhuk2023trotter}.\\

\begin{figure}
\centering
\includegraphics[width=\columnwidth]{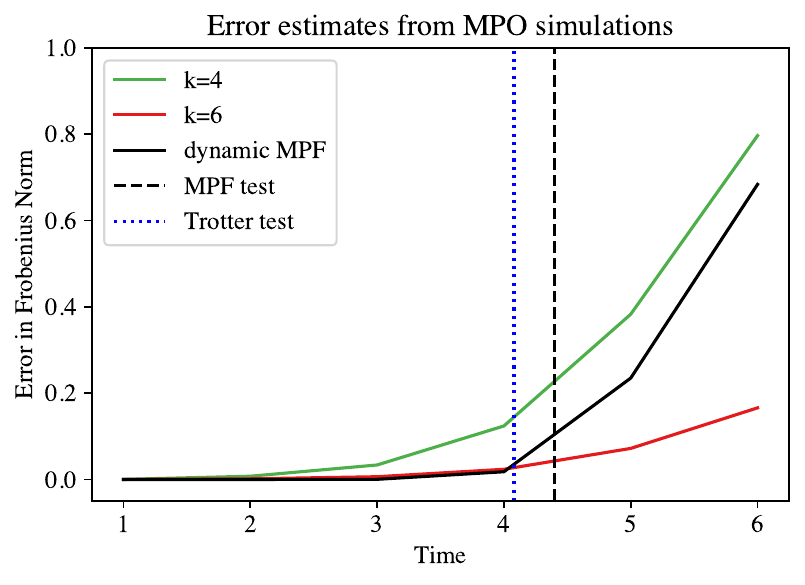}
\includegraphics[width=\columnwidth]{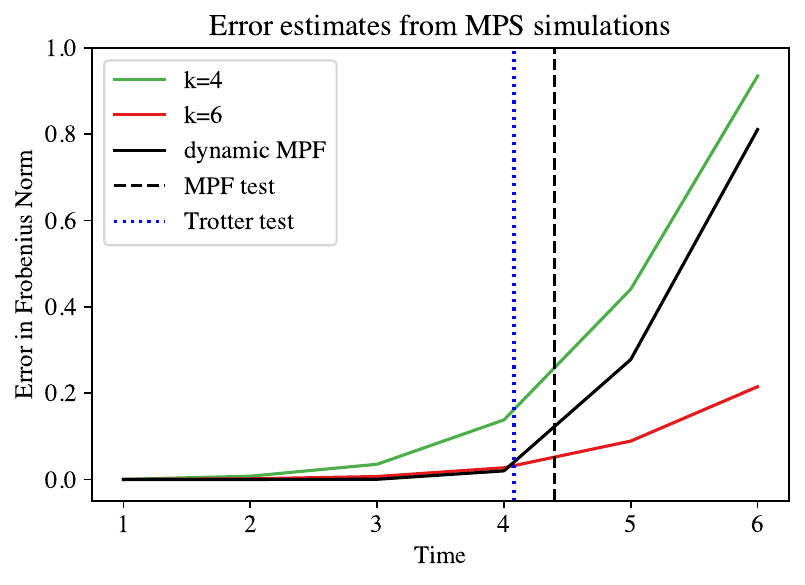}
\caption{Green line: $E_F^{k_{max}}$ in equation (\ref{efkmax}) with $k_{max}=4$. Red line: $E_F^{k_{max}}$ in equation (\ref{efkmax}) with $k_{max}=6$. Black line: $E_F^D$ in equation (\ref{proj_error}) with $k_1=2$, $k_2=3$ and $k_3=4$. Dotted blue vertical line: the time at which the error on the dynamic MPF formula with $k_1=2$, $k_2=3$ , $k_3=4$ becomes larger than the Trotter formula with $k=6$. Dashed black vertical line: the time at which the inequality in (\ref{r_plus_1_test}) no longer holds. In the top panel, all quantities were evalulated by generating the relevant MPOs using Algorithm \ref{alg:MPO} followed by the application of equations (\ref{mij}), (\ref{lj}) and ({\ref{efkF}}). In the bottom panel, all quantities were calculated by explicitly storing the relevant quantum states as Matrix Product States with $\chi=400$ and calculating their overlaps - see main text for discussion.
}\label{trott_errors_mpo_mps}
\end{figure}

\noindent We have so far compared the errors in Frobenius norm of the density matrices corresponding to one Trotter product formula vs MPF, but we should also consider errors measured in 1-norms, i.e. $\|\rho-\mu^D\|_1$ and $\|\rho-\rho_{k}\|_1$, as these quantities more closely capture the errors on the expectation values of the observables. Indeed, $\|\rho-\mu^D\|_1=\max_{\calO:\|\calO\|=1}|\Tr(\calO(\rho-\mu^D))|$, i.e. 1-norm represents the worst-case error of all possible expected values $|\Tr(\calO(\rho-\mu^D))|$. We have the following inequalities:
\begin{equation}
    \|\rho-\mu^D\|_F\le\|\rho-\mu^D\|_1\le\sqrt{r+1}\|\rho-\mu^D\|_F
\end{equation}
 where $r$ is the rank of $\mu^D$, i.e. the number of terms in the multiproduct formula $\mu^D$. Similarly, we have:
 \begin{equation}
     \|\rho-\rho_{k_j}\|_F\le\|\rho-\rho_{k_j}\|_1\le\sqrt{2}\|\rho-\rho_{k_j}\|_F
 \end{equation}
We can guarantee that the one-norm of the multiproduct error is lower than the one norm of the Trotter error, i.e. $\|\rho-\mu^D\|_1 \le \|\rho-\rho_{k_j}\|_1$ if the following inequality holds:
\begin{equation}\label{r_plus_1_test}
    \sqrt{r+1}\|\rho-\mu^D\|_F \le \|\rho-\rho_{k_j}\|_F
\end{equation}
The dashed black vertical line labelled ``MPF test" in Figure \ref{trott_errors_mpo_mps} marks the time beyond which the inequality in (\ref{r_plus_1_test}) no longer holds. This does not mean that we can only apply MPF in the allowed time regime where (\ref{r_plus_1_test}) holds, but that if one wants to provide a rigorous guarantee that dynamic MPF with a particular set of $k$'s will perform better than each of the individual $k_j$'s, then one must be in this allowed time regime. This test can also be used to test the validity of the static MPF coefficients when we consider times that are outside the window for which we have a rigorous guarantee on the performance of static MPF~\cite{zhuk2023trotter}.
\begin{figure}
\centering
\includegraphics[scale = 1.0, width=\columnwidth]{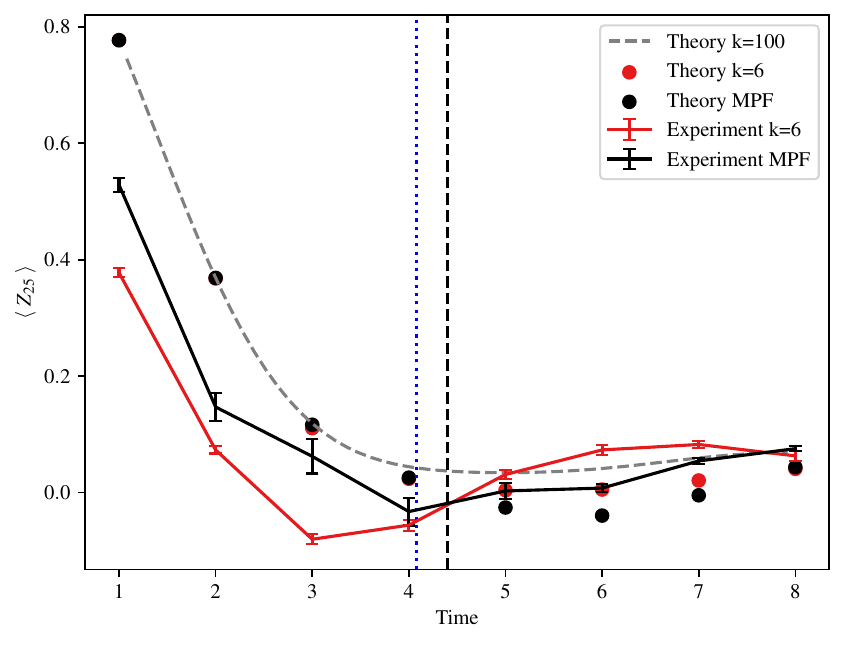}
\includegraphics[scale=1.0, width=\columnwidth]{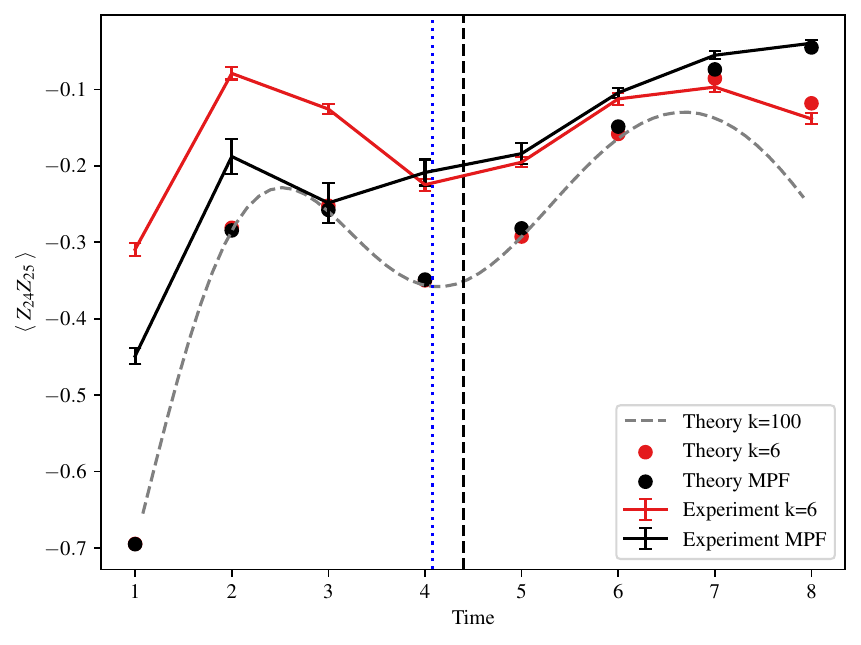}
\caption{The expectation values of the observables $\sigma^z_{25}$ and $\sigma^z_{24}\sigma^z_{25}$ after time evolution with the $50$-qubit Hamiltonian in equation (\ref{h_general}) as measured on the quantum computer $\texttt{ibm\char`_torino}$. The dynamic MPF performs better (i.e. is closer to the exact curve)  than the single Trotter circuit at times before the ``Trotter-test" and ``MPF-test" (i.e. dotted blue and dashed black vertical lines respectively as in Figure \ref{trott_errors_mpo_mps}). This is due to the reduced depths of the circuits used in the dynamic MPF.}\label{fig:heron_expec}
\end{figure}

\paragraph{Evaluation on quantum simulations.} Now we turn to the results of the MPO-MPF algorithm on quantum hardware (see Appendix~\ref{app:device_spec} for the details of the configuration). In Figure \ref{fig:heron_expec}, we plot the expectation values of two observables in the middle of the $50$-qubit spin chain, namely $\sigma^z_{25}$ and $\sigma^z_{24}\sigma^z_{25}$ respectively, as measured on the quantum computer \texttt{ibm\char`_torino}. The quantum device experiments utilize dynamical decoupling to suppress non-Markovian and crosstalk errors~\cite{pokharel2018demonstration,ezzell2023dynamical}, incorporate Pauli twirling to suppress coherent errors~\cite{vidal2003efficient}, and twirled readout extinction~\cite{van2022model} for measurement errors. See appendix \ref{sec:mpf_pea} for additional results with probabilistic error amplification (PEA) implemented on \texttt{ibm\char`_kyiv}. As above we compare a $k=6$ Trotter circuit with a dynamic multiproduct formula with $k_1=2$, $k_2=3$ and $k_3=4$. These values of $k_1, k_2, k_3$ and $k$ where chosen to ensure that the algorithmic error of the dynamic multiproduct is comparable to the algorithmic error of the Trotter circuit $k=6$; as shown in Figure \ref{trott_errors_mpo_mps}  this is the case up until time $t=4.1$ for the circuits in question. By keeping the algorithmic errors roughly equal for the Trotter circuit and the dynamic multiproduct, we can accurately assess the reduction in the device error arising from the reduced circuit depth of the circuits used in the dynamic multiproduct formula as compared to a single Trotter circuit. As shown in Figure \ref{fig:heron_expec}, the dynamic MPF results are more accurate than the Trotter results up to the ``Trotter test" cut-off time when implemented on quantum hardware. \\

We now consider the combination of our dynamic MPF algorithm with the $AQCtensor$ algorithm outlined in section \ref{sec:aqctensor}. In Figure \ref{fig:AQC_observables}, we consider the same quantities as measured on the same quantum device as in Figure \ref{fig:heron_expec} but where both curves have been enhanced by the $AQCtensor$ algorithm.  More precisely, for each value of time $t$ in Figure \ref{fig:AQC_observables} we consider two windows of time $t_1$ and $t_2$ such that $t=t_1+t_2$. The $AQCtensor$ algorithm finds a circuit of the same depth as a $k=2$ Trotter circuit that approximates the exact time evolution $e^{-iHt}$ up until a time $t_1$ quasi-exactly, i.e. with $~0.99$ fidelity. A standard Trotter circuit with either $k=2$ or $k=1$ that approximates the time evolution operator for a time window of length $t_2$ is then appended to the optimized circuit to simulate times $t$ that are larger than $t_1$. In this instance, we take $t_1=3$ and hence $t_2=t-3$.  In general, the value of $t_1$ can be taken to be as large as the available classical resources can simulate the quantum system with a given desired precision. For each value of $t$ in Figure \ref{fig:AQC_observables} we have two quantum circuits, one with the $t_2$ time window divided into $k=2$ Trotter steps and one where $t_2$ is approximated with $k=1$ Trotter steps. The red curve represents the results of the first of these circuits  and the black curve represents the result of combining the two circuits with a dynamic multiproduct. In the $AQCtensor$ workflow, the observables can be calculated quasi-exactly using the classical MPS based method for times up until $t=t_1$, hence we do not include results from quantum hardware for $t<t_1$. To calculate the dynamic multiproduct coefficients $c_j(t)$ for these $AQCtensor$ circuits we use Algorithm \ref{alg:MPO} and the method outlined in \ref{sec:mpo-mpf}, but with a slight modification of the quantity $F_{ij}$ in equation (\ref{F_def}). In particular, since we now divide each time $t$ into two time windows $t_1$ and $t_2$ where for times $t<t_1$ the system is simulated quasi-exactly using the $AQCtensor$ algorithm, we must define the quantities $F^{AQC}_{ij}$ and $F^{AQC}_{\rm ex,j}$:
\begin{equation}\label{eq:Faqc_defs}
\begin{aligned}
    F^{AQC}_{ij} & \equiv e^{iHt_1}S\left(\frac{t_2}{k_i}\right)^{-k_i}S\left(\frac{t_2}{k_j}\right)^{k_j}e^{-iHt_1}\\
    F^{AQC}_{\rm ex,j} & \equiv e^{iH(t_1+t_2)}S\left(\frac{t_2}{k_j}\right)^{k_j}e^{-iHt_1}\\
\end{aligned}
\end{equation}
The quantities $M_{ij}$ and $L_j$ are then calculated as described in section \ref{sec:mpo-mpf} and the full \emph{MPO-MPF} algorithm proceeds as previously described. Finally, we comment on the ``MPF test" from equation (\ref{r_plus_1_test}) in the context of the $AQCtensor$+dynamic MPF combination. Since the application of $AQCtensor$  evolves the system almost exactly for a time $t_1$, one can expect both the Trotter errors and MPF errors to be reduced and for the inequality in (\ref{r_plus_1_test}) to hold for longer total simulation times $t$. In Figure \ref{fig:AQC_observables} we show that this is indeed the case; the vertical line representing the final time at which the inequality in (\ref{r_plus_1_test}) holds occurs later than in Figure \ref{fig:heron_expec}. Thus one of the key advantages of using $AQCtensor$ is that it allows for the application of our \emph{MPO-MPF} algorithm at times later than would otherwise be feasible.
\begin{figure}
\centering
\includegraphics[width=\columnwidth]{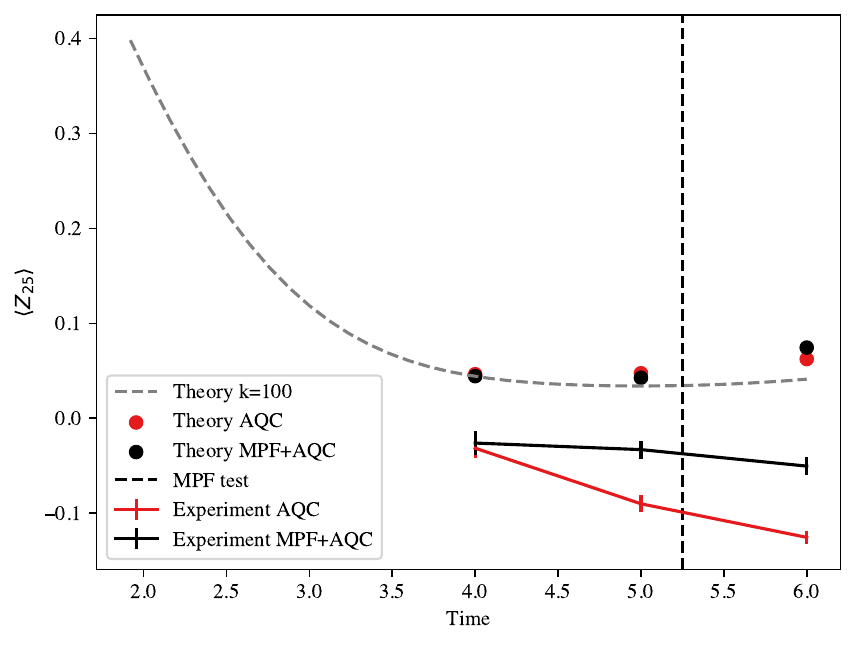}
\includegraphics[width=\columnwidth]{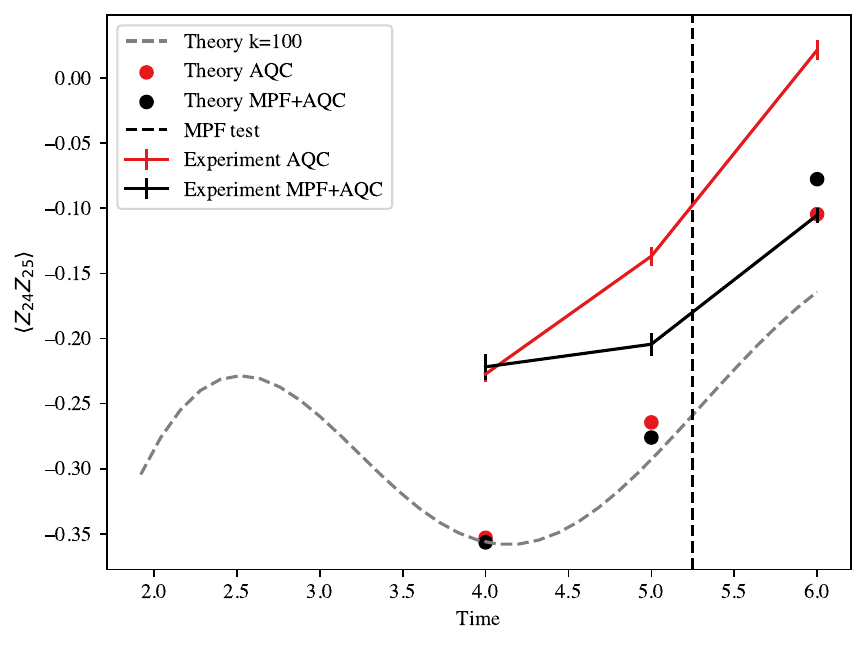}
\caption{The expectation values of the same observables considered in Figure \ref{fig:heron_expec}  measured on \texttt{ibm\char`_torino} but where the circuits have been enhanced by the $AQCtensor$ algorithm. In particular, the system is simulated almost exactly using a classical MPS-based algorithm up to time $t_1=3$ and an optimized quantum circuit is found which closely approximates this exact evolution up to time $t_1$. This optimized circuit has the same depth and same structure as a Trotter circuit with $k=2$. Additional Trotter circuits with $k=2$ and $k=1$ are appended to the optimized circuit to simulate times beyond $t=3$. The red curve represents the results of the $AQCtensor$ + $k=2$ Trotter circuit and the black curve represents the result of combining the two circuits with our dynamic multiproduct algorithm. The dashed black vertical line represents the ``MPF-test" which occurs at a later time than in Figure \ref{fig:heron_expec} due to the effect of $AQCtensor$ - see main text.}\label{fig:AQC_observables}
\end{figure}

\section{Discussion}\label{sec:discussion}
In this work, we considered the question of whether one can apply state of the art classical Tensor Network algorithms in tandem with quantum simulations to provide an advantage over applying either method in isolation. We addressed this question by introducing a novel dynamic multiproduct algorithm based on Matrix Product Operators and by combining it with the previously proposed $AQCtensor$ algorithm \cite{robertson2023approximate}. The error analysis in Figure \ref{bootstrap} and in Appendix \ref{sec:error-analysis} demonstrates that our algorithm does indeed provide an advantage over a purely quantum or purely classical approach. Furthemore, we demonstrated our algorithm on the quantum processors \texttt{ibm\char`_torino} and \texttt{ibm\char`_kyiv} (see appendix~\ref{sec:mpf_pea}) and showed how our algorithm allowed for the reduction in circuit depth and hence a reduction in circuit errors - see e.g. Figures \ref{fig:heron_expec} and \ref{fig:kyiv_expec}. We anticipate that this work may inspire similar Tensor Network + quantum algorithms for applications beyond that of the simulation of a $1D$-spin chain considered here. Furthermore, we also briefly comment on how the algorithm presented here could be applied to more complex models, e.g. in higher dimensions. While Matrix Product State based methods are not well suited to study models in dimensions larger than $1$, we note that our algorithm never explicitly stores the quantum state itself but only the Matrix Product Operator representing the object $F$ defined in equation~\eqref{eq:Fdefs}. The theoretical error analysis in appendix \ref{sec:error-analysis} should still apply beyond 1D - more numerical work is needed to determine how well the \emph{MPO-MPF} method would work in practice in $2D$ or if a more sophisticated Tensor Network approach that goes beyond Matrix Product Operators would be required.

\begin{acknowledgments}
The authors would like to thank Antonio Mezzacapo for useful discussions, guidance and coordination of this project. The authors thank Nicolas Lorente, Kate Marshall and Max Rossmannek for useful discussions.
\end{acknowledgments}

\clearpage
\appendix
\section{Error Analysis}\label{sec:error-analysis}

In section \ref{sec:mpo-mpf}, the precision and memory requirements to store $F_{ij}$ in Algorithm \ref{alg:MPO}
were briefly discussed - we elaborate on these points here. Recall that to find the dynamic MPF coefficients $c_j(t)$ we minimize $E_F^D$ in equation (\ref{proj_error}) subject to the constraint $\sum\limits_i c_i = 1$. To do so, we must calculate the quantities $M_{ij}$ and $L_j$ defined in equations (\ref{eq:Mt}) and (\ref{eq:Lexact}), which can be rewritten in terms of $F_{ij}$ and $F_{ex,j}$ - see equation (\ref{eq:Fdefs}). Recall that $F_{0j} \equiv S^{\dagger}\left(\frac{t}{k_0}\right)^{k_0} S\left(\frac{t}{k_j}\right)^{k_j}  \approx e^{iHt} S\left(\frac{t}{k_j}\right)^{k_j}$ with $k_0 >> k_j$. In what follows, we thus consider the object:
\begin{equation}\label{F_def}
    F \equiv e^{iHt} S\left(\frac{t}{k}\right)^{k}
\end{equation}
and we take the Hamiltonian in equation (\ref{h_general}) in the main text. We consider the memory requirements to generate and store $F$ as a Matrix Product Operator and compare these to the requirements to store the quantum state $\ket{\psi_t} = e^{-iHt} \ket{\psi_0}$ as a Matrix Product State. In classical time evolution algorithms such as TEBD one stores $\ket{\psi_t}$ as an MPS by repeated applications of SVD and by truncating the representation to include only singular values that are larger than a target precision threshold $\lambda_0$. The bond dimension required to store the state with precision $\lambda_0$ increases exponentially with time $t$:
\begin{equation}\label{chi_mps_scaling}
    \chi_{mps}(\lambda_0) =  f(\lambda_0) e^{v_0 kdt}  = f(\lambda_0) e^{v_0 t} 
\end{equation}
where $dt=\frac{t}{k}$. We demonstrate this behaviour in Figure \ref{chi_mps_vs_t}. In the top panel, we plot $\log\chi_{mps}$ vs $t$ for a range of values of $\lambda_0$ and find an estimate for $f(\lambda_0)$ and $v_0$ by fitting a linear function to each curve. In the bottom panel, we plot these estimates for $f(\lambda_0)$ and $v_0$ vs $\lambda_0$ and find that, as expected, $f(\lambda_0)$ is strongly dependent on $\lambda_0$ whereas $v_0$ is approximately constant.\\

\begin{figure}
\centering
\includegraphics[width=\columnwidth]{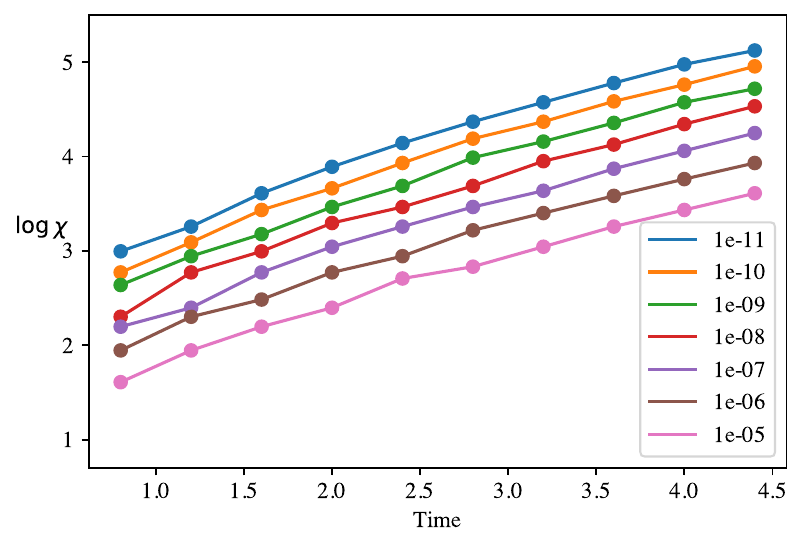}
\includegraphics[width=\columnwidth]{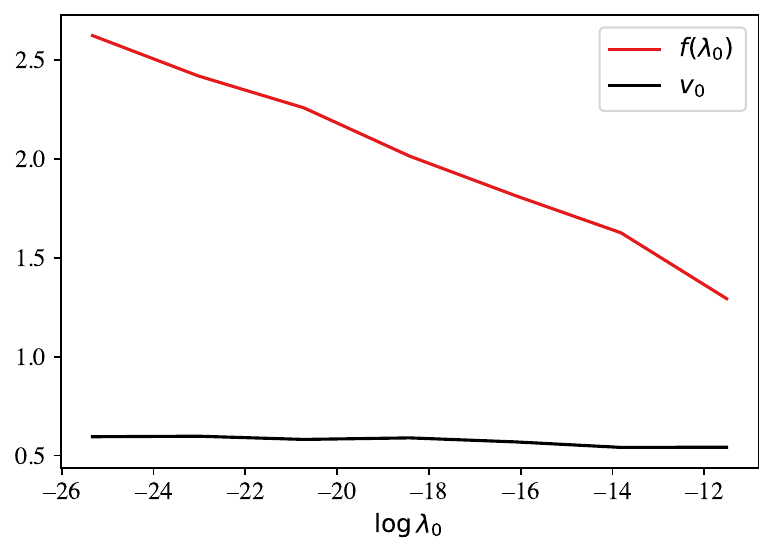}
\caption{Top panel: demonstration of the scaling behaviour of the bond dimension $\chi_{mps}$ in equation (\ref{chi_mps_scaling}). $\log\chi_{mps}$ is plotted  vs time for a range of values of the cut off $\lambda_0$. Bottom panel: the quantities $f(\lambda_0)$ and $v_0$ are found by fitting a linear function to each of the curves in the top panel. As expected, one sees that $f$ strongly depends on $\lambda_0$ while $v_0$ does not.}\label{chi_mps_vs_t}
\end{figure}

\noindent We expect a similar behaviour for the bond dimension required to store the Matrix Product Operator representation of the full time evolution operator $e^{-iHt}$:
\begin{equation}\label{chi_mpo_scaling_ex}
\chi_{mpo}(\lambda_0) = g(\lambda_0) e^{v_1kdt} = g(\lambda_0) e^{v_1 t}
\end{equation}
Now instead of the unitary $e^{-iHt}$, we consider the bond dimension required to store the object $F$ in (\ref{F_def}) with precision $\lambda_0$.  We expand $F$ using the BCH formula to get:
\begin{equation}\label{F_BCH}
    F \approx e^{-idt^3 k C} = e^{-i t dt^2 C}
\end{equation}
where $C$ is given by a sum of nested commutators. Equation (\ref{F_BCH}) suggests that one can consider $F$ as the time evolution operator with effective Hamiltonian $H_{eff}=Cdt^2$. In place of equation (\ref{chi_mpo_scaling_ex}), we thus write:
\begin{equation}\label{chi_mpo_scaling}
\chi_{mpo}^F(\lambda_0) = g(\lambda_0) e^{v_1kdt^{\alpha}} = g(\lambda_0) e^{v_1\frac{t^{\alpha}}{k^{\alpha-1}}}
\end{equation}
where $\alpha$ is some constant that should be lower-bounded by $3$. We test this scaling behaviour in Figures \ref{chi_k} and \ref{chi_dt}. In Figure \ref{chi_k}, we plot $\log\log(\chi_{mpo}^F)$ vs $\log k$ with $t=4$ and find that $\alpha\approx 4.7$. Similarly, in Figure \ref{chi_dt} we plot $\log(\chi_{mpo}^F)$ vs time for fixed $dt$. Equations (\ref{chi_mpo_scaling}) and (\ref{chi_mps_scaling}) thus show that for any desired precision $\lambda_0$ of the quantities $M_{ij}$ and $L_j$ in equation (\ref{proj_error}), there exists a value of $k$ above which $\chi_{mpo}^F(\lambda_0) < c \chi_{mps}(\lambda_0) $ for any arbitrary constant $c$ and hence one gains an advantage from using the MPO-MPF algorithm outlined in section \ref{sec:mpo-mpf} over using either a purely classical MPS-based approach or a purely quantum Trotterization-based approach.
\begin{figure}
\centering
\includegraphics[width=\columnwidth]{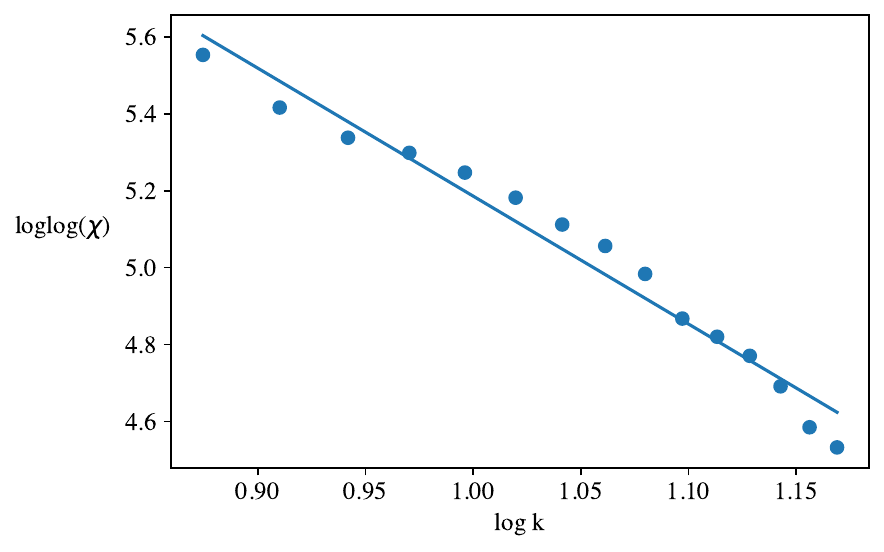}
\caption{A test of the scaling of $\chi_{mpo}$ with $k$ for fixed $T$ - see equation (\ref{chi_mpo_scaling}). The slope of the line is $-3.7$, suggesting that $\alpha\approx4.7$ in equation (\ref{chi_mpo_scaling}).}\label{chi_k}
\end{figure}

\begin{figure}
\centering
\includegraphics[width=\columnwidth]{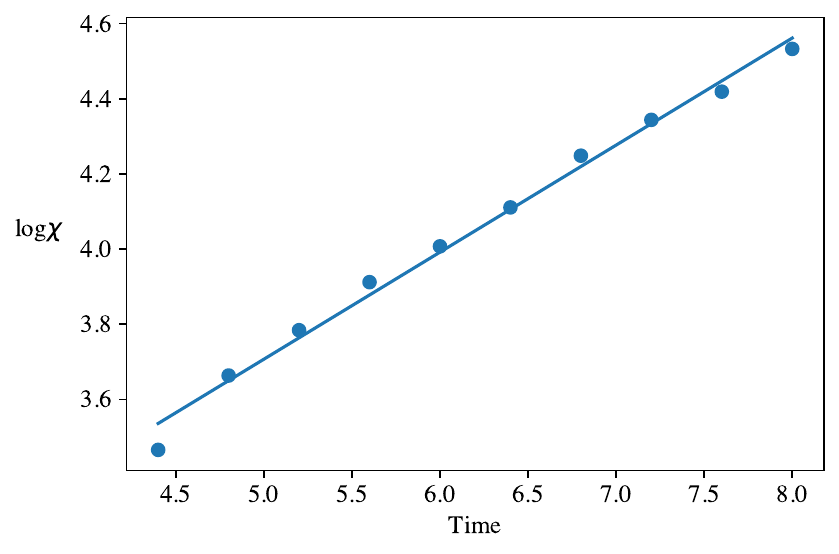}
\caption{A test of the scaling of $\chi_{mpo}$ with $k$ for fixed $dt$.}\label{chi_dt}
\end{figure}
\paragraph{MPO truncation error for MPF coefficients.} We now discuss the effect of the truncation error $\lambda_0$ on the MPF coefficients $c_j$ of the dynamic multiproduct formula $\mu^D=\sum_{j=1}^r c_j\rho_{k_j}$ and on the MPF observables $\sum_{j=1}^r c_j \mathcal{O}_j$, $\calO_j=\operatorname{Tr}(O\rho_{k_j}(t))$.\\ 
Let $\mpoM$ and $\mpoL$ denote MPO-approximations of the exact matrix $M$ and vector $L$ defined in~\eqref{eq:Lexact}. Let $\mpoMerr=M^{(\lambda_0)}-M$ be the matrix of MPO-approximation errors, and set $\mpoLerr=\mpoL-L$. Let $\oc$ denote the vector of exact dynamic MPF coefficients, the unique minimizer of~\eqref{proj_error} with exact $M$ and $L$, and let $\mpoC$ denote the unique minimizer of~\eqref{proj_error} with $M$ and $L$ substituted by $\mpoM$ and $\mpoL$ respectively. Clearly, $\oc$ can be computed by solving the linear equation $M \oc=L+\mu \one$ where Lagrange multiplier $\mu$ is chosen so that $\la\one|\oc\ra=1$: $\oc = M^{-1} L + \mu M^{-1} \one$ and $\mu^\star=(1-\la\one|M^{-1} |L\ra)/ \la\one|M^{-1}| \one\ra$. 

% Suppose $M\in \RR^{r\times r}$ is the Gram matrix describing overlaps between MPF basis functions and $L \in \RR^r$ is a vector of overlaps
% between MPF basis functions and the exact time evolved state. 
% As shown in Sergiy's proof, the optimal  vector of MPF coefficients $c\in \RR^r$ obeys
% \begin{equation}
% \label{eq1}
% M|c\ra =|L\ra + \mu |\one\ra,
% \end{equation}
% where $\mu\in \RR$ is a Lagrange multiplier chosen such that $\la \one|c\ra=1$.
% Suppose $\tilde{M}\in \RR^{r\times r}$ is an approximation of $M$
% and $\tilde{L}\in \RR^r$ is an approximation of $L$ computed by MPO. The corresponding vector of approximate MPF coefficients $\tilde{c}\in \RR^r$ obeys
% \begin{equation}
% \label{eq2}
% \tilde{M}|\tilde{c}\ra =|\tilde{L}\ra + \tilde{\mu} |\one\ra,
% \end{equation}
% where $\tilde{\mu}\in \RR$ is chosen such that $\la \one|\tilde{c}\ra=1$.
% Our goal is to upper bound the error $\|c-\tilde{c}\|$. We have 
\begin{lemma}\label{lemma:mpo_error}
    Let $\mpoC=(\mpoM)^{-1}(\mpoL + \mu\one)$ and $\mu$ is chosen to satisfy $\la\one|\mpoC\ra=1$. Then for any $1>\varepsilon>0$ there exist $\lambda_0$ such that $\|L-\mpoL\|_2,\|M-\mpoM\|<\varepsilon<1$ and     
    \begin{equation}\label{eq:mpf-mpo-error-sb}
\|\oc-\mpoC\|_2\le \varepsilon \frac{  \|M^{-1} \| + \|\oc\|_2 }{1-\varepsilon},
\end{equation}
For $\calO=(\calO_1\dots\calO_r)^\top$ let $\calE(t)=O^\star\one-\calO$ be the vector of Trotter errors, where $O^\star=\operatorname{Tr}(O\rho(t))$, $\rho$ -- the exact density matrix. Then 
    \begin{equation}\label{eq:mpf-mpo-error-obs}
        |\la\mpoC|\calO\ra-O^\star|\le |\la\oc|\calE\ra|+\varepsilon \frac{  \|M^{-1} \| + \|\oc\|_2 }{1-\varepsilon}\|\calE\|_2
    \end{equation}  
\end{lemma}
\begin{proof}
Using the notion of the square-root of a symmetric positive semidefinite matrix, namely $M^{1/2}M^{1/2}=M$ it is easy to see that:  
\begin{equation}
\begin{split} \label{eq2'}
\| \oc-\mpoC\|^2 &=  \la \oc-\mpoC|M^{1/2} M^{-1} M^{1/2} |\oc-\mpoC \ra\\
&\le \| M^{-1}\|\cdot \la \oc-\mpoC|M|\oc-\mpoC\ra.
\end{split}
\end{equation}
Let us upper bound $\la \oc-\mpoC|M|\oc-\mpoC\ra$. Since $\oc = M^{-1} L + \mu M^{-1} \one$ and $\la \oc-\mpoC|\one\ra=0$ it follows that: 
\begin{align*}
  \la \oc&-\mpoC|M|\oc-\mpoC\ra = \la \oc-\mpoC|L\ra + \mu \la \oc-\mpoC|\one\ra\\
  &- \la \oc-\mpoC|M|\mpoC\ra= \la \oc-\mpoC|L\ra - \la \oc-\mpoC|M|\mpoC\ra   
\end{align*}
Substituting $M=\mpoM+M-\mpoM$ in $\la \oc-\mpoC|M|\mpoC\ra$ and recalling that $\mpoM\mpoC=\mpoL + \mu\one$ with $\mu$ chosen to satisfy $\la\one|\mpoC\ra=1$ we get: 
\begin{align*}
    \label{eq3}
\la \oc&-\mpoC|M|\oc-\mpoC\ra =  \la \oc-\mpoC|L-\mpoL\ra\\
&- \mu\la \oc-\mpoC|\one\ra - \la \oc-\mpoC|M-\mpoM|\mpoC\ra\\
&=\la \oc-\mpoC|L-\mpoL\ra - \la \oc-\mpoC|M-\mpoM|\mpoC\ra\\
&\le \|\oc-\mpoC\|_2(\|L-\mpoL\|_2+\|M-\mpoM\|\|\mpoC\|_2)\\
&\le \|\oc-\mpoC\|_2(\|L-\mpoL\|_2+\\
& \ \ \ \ \  \|M-\mpoM\|\|\oc-\mpoC\|_2\|\oc\|_2)
\end{align*}
Substituting this bound into the r.h.s. of~\eqref{eq2'}, dividing the result by $\|\oc-\mpoC\|_2$ and isolating $\|\oc-\mpoC\|_2$ on the l.h.s. of the resulting inequality we obtain~\eqref{eq:mpf-mpo-error-sb} provided $\|M-\mpoM\|<1$. 

Now, \eqref{eq:mpf-mpo-error-obs} follows by using Cauchy inequality to upper-bound $|\la\mpoC|\calO\ra-\la\mpoC|O^\star\one\ra|$:\[
|\la\mpoC\pm\oc|\calO\ra-\la\mpoC|O^\star\one\ra|\le \|\mpoC-\oc\|_2\|\calE\|_2+|\la\oc|\calE\ra|
\]
followed by~\eqref{eq:mpf-mpo-error-sb}. This concludes the proof. 
\end{proof}

%Since $\mpoLerr$, $(\mu-\mu^{(\lambda_0)})$ and the term within $\left(\cdot\right)$ can be made arbitrarily small for an appropriate choice of $\lambda_0$ it follows that $\forall\varepsilon>0$ there exits $\lambda_0$ such that $\|\oc-\mpoC\|<\varepsilon$. 

%We will express $(M^{(\lambda_0)})^{-1}$ as a function of $\Psi_{ij}$ to show that reducing $\lambda_0$ we can control In what follows we will use linear indexing, namely we 
%\Psi^{k,j+1} = \Psi^{k,j}+\Psi_{k(j+1)}|\ort_{k}\ra\la\ort_j|
%$$, $M^{(\lambda_0)}=M_r$ provided $M_r=M_{r-1} + \varepsilon_r\,|u_r\ra\la v_r|$  can be represented as Assuming that MPO error of storing $F_{ij}$ 

% \section{IBM Kyiv results with PEA and MPF}

% \begin{figure}
% \centering
% \includegraphics[width=\columnwidth]{z_25_25_kyiv_mpf.pdf}
% \caption{}\label{z25z25_kyiv}
% \end{figure}

% \begin{figure}
% \centering
% \includegraphics[width=\columnwidth]{z_40_41_kyiv_mpf.pdf}
% \caption{}\label{z40z41_kyiv}
% \end{figure}
\section{Probabilistic Error Amplification}
\label{sec:mpf_pea}

\begin{figure}
\centering
\includegraphics[width=\columnwidth]{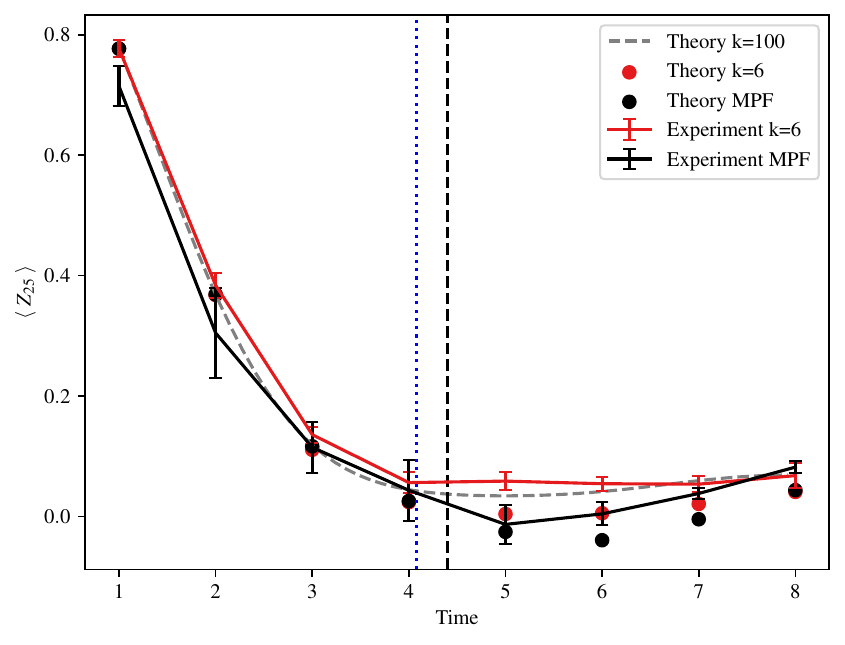}
\includegraphics[width=\columnwidth]{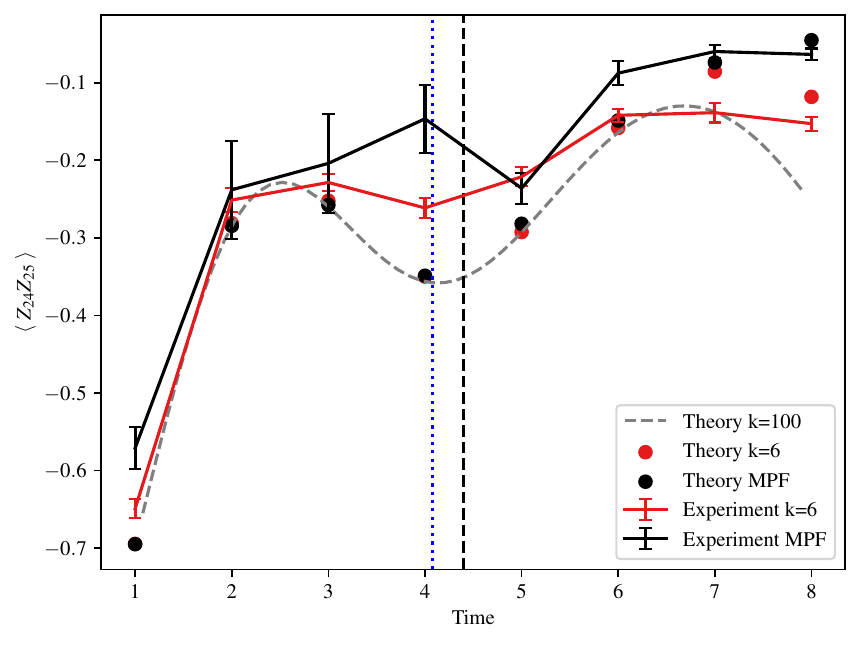}
\caption{The expectation values of the observables $\sigma^z_{25}$ and $\sigma^z_{24}\sigma^z_{25}$ after time evolution with the Hamiltonian in equation (\ref{h_general}) as measured on $\texttt{ibm\char`_kyiv}$. Both MPF and $k=6$ results are shown after incorporating probabilistic error amplification. Dynamic MPF and $k=6$ results both give comparable results even though $k=6$ circuit is considerly deeper than the circuits used for the MPF.} 
\label{fig:kyiv_expec}
\end{figure}

Before any error mitigation and noise model simplification, the general noise channel has the Kraus operator sum representation (KOSR)
$$\tilde \Lambda(\rho) = \sum_{\alpha, \beta} c_{\alpha, \beta} P_{\alpha} \rho P_{\beta}^{\dagger} $$
where $P_{\alpha} \in \mathcal{P}_{n}$ the n-qubit Pauli group. 

Performing Pauli twirling over the Pauli group, which can be easily implemented as single qubit gates results in
$$ \Lambda(\rho) = \mathbb{E}_\alpha \left[P_\alpha^{\dagger} \tilde{\Lambda}\left(P_\alpha^{\dagger} \rho P_\alpha\right) P_\alpha^{\dagger}\right] = \sum_{\alpha} c_{\alpha} P_{\alpha} \rho P_{\alpha} ^{\dagger}$$
Owing to the Pauli twirling, the Pauli transfer matrix becomes diagonal, i.e., it goes from 
$T_{\tilde{\Lambda}}[a, b]=\frac{1}{2^n} \operatorname{Tr}\left[P_a^{\dagger}\left(\tilde{\Lambda}\left(P_b\right)\right)\right]$ to $T_{\Lambda}[a, b]=\delta_{a, b} T_{\tilde{\Lambda}}[a, b]$. We supplement Pauli twirling with dynamical decoupling (DD), in particular the CPMG sequence -  $f_{\tau/2} X_p f_{\tau} X_m f_{\tau/2}$ sequence. Here $f_{\tau}$ signifies free evolution for time $\tau$ and $X_p$, $X_m$ are $\pm\pi/2$ rotation about the $X$-axis. The CPMG sequence cancels $Z$-terms in the interaction Hamiltonian and more pertinently, when applied on alternating qubits, cancels the $ZZ$-crosstalk prevalent in superconducting qubits. DD combined with Pauli twirling suppresses coherent, crosstalk and non-Markovian 
errors. In turn, this allows us to model the device noise as a sparse Pauli-Lindblad model of the form

$$ \Lambda(\rho) = \prod_{k \in \mathcal{K}} \left( \frac{e^{-2 \lambda_{k}}+1}{2} \rho + \frac{e^{-2 \lambda_{k}}-1}{2} P_{k} \rho P_{k}^{\dagger}\right).$$
Here $\mathcal{K}$ correspond to single and two-qubit Pauli terms on the qubits on which gates are being actively applied~\cite{kim2023evidence}. 

The first step in PEA is to learn the model coefficients $\lambda_k$.  Given the circuit for Trotter evolution, we identify unique gate layers and then perform benchmark circuit for these layers at different depths to empirically determine $\lambda_k$s. Equipped with the model coefficients, we can then execute the circuit at various noise scales $\alpha$. Multiplying the noise coefficients $\lambda_k$ by $\alpha$ gives a Pauli map with the noise scaled by $\alpha$, which in turn allows us to use zero noise extrapolation to find the ideal values of the desired observable. It is necessary to perform randomized twirling and dynamical decoupling through both the learning and the amplification steps of PEA, as these error suppression steps enforce that noise model can be approximated using a Pauli-Lindblad model.

The results in in Figure \ref{fig:kyiv_expec} from \texttt{ibm\char`_kyiv} were acquired by combining MPF with probabilistic error amplification (PEA). As MPF focuses solely on addressing mitigating Trotter errors, any decoherence that causes the underlying product formula to deviate from the noiseless expectation values propagates as errors in the MPF results. Therefore, in theory, MPF is both complementary to and aided by error suppression and error mitigation methods which must be implemented before the product formulas are used combined.

\section{Device Specifications}\label{app:device_spec} 

\paragraph*{Device overview.}
The experimental results reported in Sec. \ref{sec:results} were obtained on two different devices, with different architectures. $\texttt{ibm\char`_kyiv}$ is a so-called Eagle type quantum processor, consisting of fixed-frequency transmon qubits with capacitive coupling between neighboring qubits arranged in a heavy-hexagonal lattice of 127 qubits. $\texttt{ibm\char`_torino}$ is a Heron type device, which consists of fixed frequency transmon qubits with flux tunable-coupling arranged in a heavy-hexagonal lattice of 133 qubits. The latter Heron type device are a newer generation devices compared to the Eagle ones, where the cross-talk effects due to the always-on capacitive coupling between neighboring qubits has been greatly reduced by allowing this coupling to be tunable and effectively ensuring no-coupling between neighbors unless an entangling interaction between them is applied. Moreover, the tunable coupling allows for stronger (thus faster) interaction between qubits, reducing the gate time of two-qubit operations in this architecture to be the same (as order of magnitude) as single qubit operations. This often results in the need for less complicated error suppression/error mitigation techniques for the experiments involving $\texttt{ibm\char`_torino}$, as one of the main contribution to the noise in Eagle type devices comes from the cross-talk terms arising in the presence of always-on coupling between qubits during entangling gates. 

\paragraph*{Device properties.}
The IBM Quantum Eagle processor's error properties are reported in Fig.~\ref{fig:kyiv_t1t2} while the Heron processor's error properties are shown in Fig.~\ref{fig:torino_t1t2}. For both figures, the left panel shows the single-qubit gate error, characterized by the randomized benchmarking technique, the two-qubit gate error, also assessed via randomized benchmarking,  and the readout error, representing the readout assignment infidelity. The right panel focuses on coherence times, the $T_1$ relaxation time, the period a qubit takes to relax to its ground state, and the $T_2$ dephasing time, measuring the time over which a qubit maintains its quantum state coherence.

\begin{figure}
\centering
\includegraphics[width=\columnwidth]{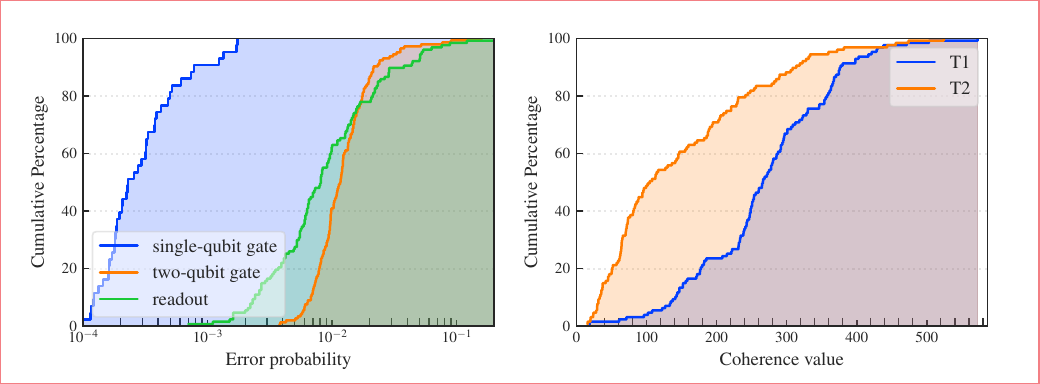}
\caption{Measured properties of $\texttt{ibm\char`_kyiv}$ device. Left panel: single-qubit gate error $.04\%(.023\%)$ mean(median), two-qubit gate error $1.5\%(1.2\%)$ mean(median), readout error $1.6\%(.81\%)$ mean(median). Right panel: T1 $270(270) \mu s$ mean(median) and T2 $150(100) \mu s$ mean(median).}\label{fig:kyiv_t1t2}
\end{figure}

\begin{figure}
\centering
\includegraphics[width=\columnwidth]{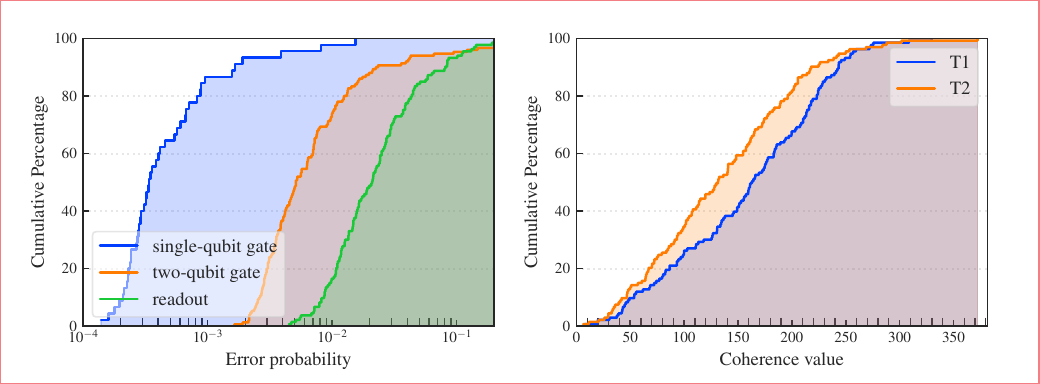}
\caption{Measured properties of $\texttt{ibm\char`_torino}$ device. Left panel: single-qubit gate error $.11\%(.034\%)$ mean(median), two-qubit gate error $4.3\%(0.51\%)$ mean(median), readout error $3.5\%(2.1\%)$ mean(median). Right panel: T1 $160(160) \mu s$ mean(median) and T2 $130(130) \mu s$ mean(median).}
% \caption{
% single-qubit gate    mean=1.1e-03  median=3.4e-04
% two-qubit gate       mean=4.3e-02  median=5.1e-03
% readout              mean=3.5e-02  median=2.1e-02
% T1                   mean=1.6e+02  median=1.6e+02
% T2                   mean=1.3e+02  median=1.3e+02}
\label{fig:torino_t1t2}
\end{figure}

\paragraph*{Qubit used in the experiments.}
The qubit subsets chosen on the different devices are presented in Figs.~\ref{fig:kyiv_connectivity} and \ref{fig:torino_connectivity} for $\texttt{ibm\char`_kyiv}$ and $\texttt{ibm\char`_torino}$, respectively. The susbsets were chosen by leveraging the native capabilities of Qiskit's transpiler. This uses the $\text{VF2}$ graph isomorphism algorithm to find subsets of device qubits whose interaction graph is isomorphic to the one of the routed (i.e. where SWAP gates are added to account for device connectivity) circuit. Isomorphic layouts are then scored using a heuristic map \cite{nation2023suppressing} based on the error rates reported for the devices. The layout with the best (lowest) score is selected as the layout for the experiment.

\begin{figure*}
\centering
\includegraphics[width=\textwidth]{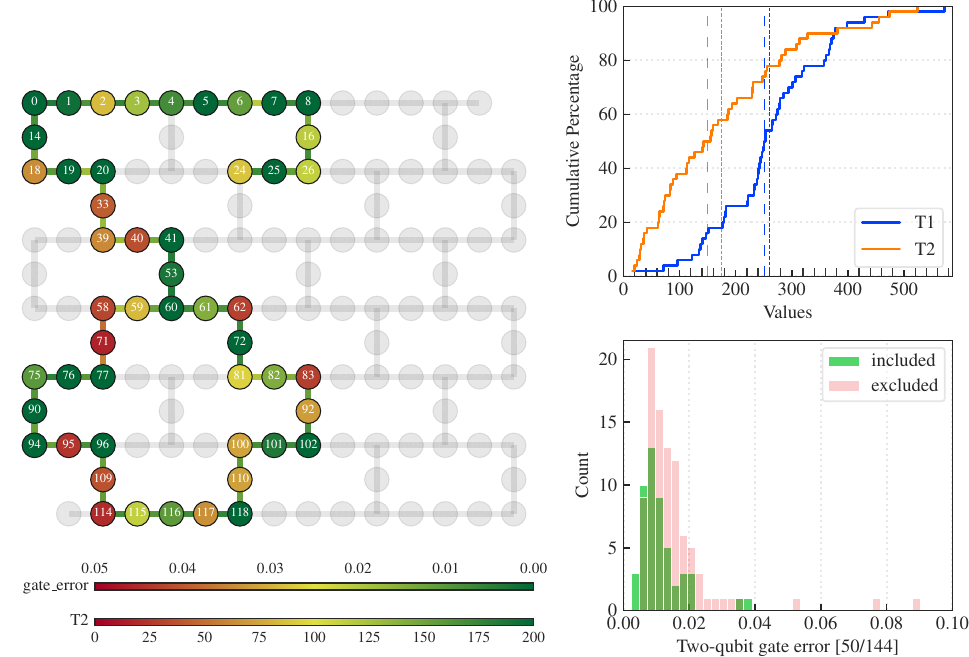}
\caption{Qubit subest used in the experiment on $\texttt{ibm\char`_kyiv}$. Left panel: Device topology with used qubits/edges highlighted according to their measured error rates ($T_2$ and two-qubit gate error). Right panel: (top) cumulative distribution of coherence times $T_1$ and $T_2$ with highlighted mean(median) values as dotted(dashed) lines, (bottom) histogram of two-qubit error rates for edges included/excluded from the chosen subset.}\label{fig:kyiv_connectivity}
\end{figure*}

\begin{figure*}
\centering
\includegraphics[width=\textwidth]{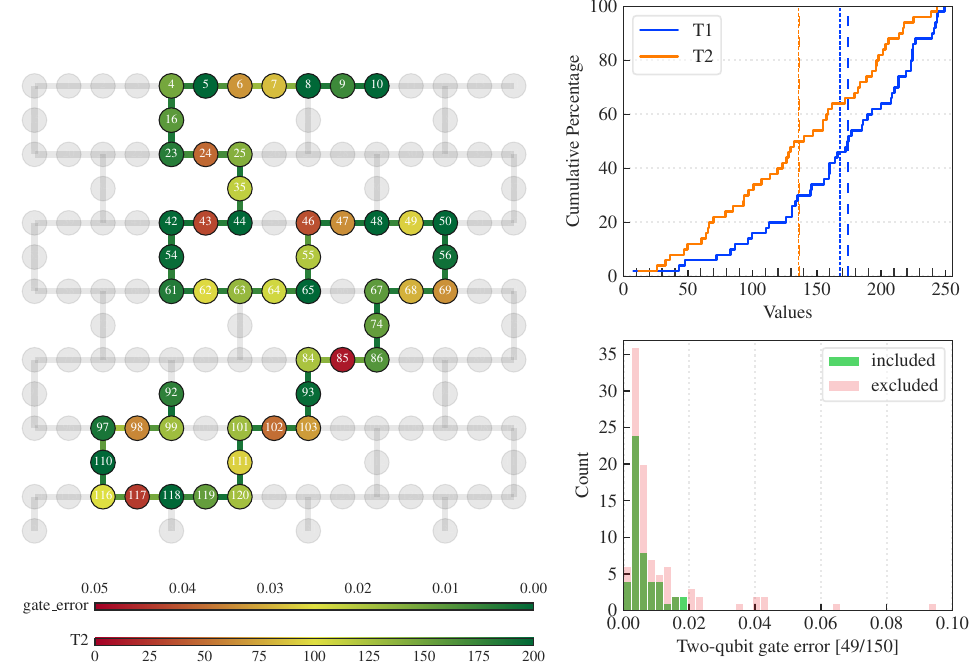}
\caption{Qubit subest used in the experiment on $\texttt{ibm\char`_torino}$. Left panel: Device topology with used qubits/edges highlighted according to their measured error rates ($T_2$ and two-qubit gate error). Right panel: (top) cumulative distribution of coherence times $T_1$ and $T_2$ with highlighted mean(median) values as dotted(dashed) lines, (bottom) histogram of two-qubit error rates for edges included/excluded from the chosen subset.}\label{fig:torino_connectivity}
\end{figure*}

\clearpage

\bibliographystyle{unsrt}
\bibliography{references}

\end{document}